\documentclass[12pt, a4paper]{article}
\usepackage[utf8]{inputenc}
\usepackage{tensor}
\usepackage[margin=0.8in]{geometry}
\usepackage{mathtools}
\usepackage{authblk}
\usepackage{amssymb}
\usepackage{physics}
\usepackage{url}
\usepackage{amsmath}
\usepackage{appendix}
\usepackage{hyperref}
\usepackage{caption}
\usepackage{amsthm}
\usepackage{stmaryrd}
\usepackage{bigints}
\newtheorem{theorem}{Theorem}[section]
\newtheorem{corollary}{Corollary}[theorem]
\newtheorem{lemma}[theorem]{Lemma}
\newtheorem{definition}[theorem]{Definition}

\newcommand{\overbar}[1]{\mkern 1.8mu\overline{\mkern-1.8mu#1\mkern-1.8mu}\mkern 1.8mu}
\usepackage{graphicx}
\numberwithin{equation}{section}

\title{\textbf{Spinorial quasilocal mass for spacetimes with negative cosmological constant}}
\author{Virinchi Rallabhandi\thanks{v.v.rallabhandi@sms.ed.ac.uk}}
\affil{School of Mathematics and Maxwell Institute for Mathematical Sciences, University of Edinburgh, King's Buildings, Edinburgh, UK, EH9 3FD}
\date{17th of April, 2025}

\begin{document}

\maketitle

\begin{abstract}
\noindent A new notion of quasilocal mass is defined for generic, compact, two dimensional, spacelike surfaces in four dimensional spacetimes with negative cosmological constant. The definition is spinorial and based on work for vanishing cosmological constant by Penrose and Dougan \& Mason. Furthermore, this mass is non-negative, equal to the Misner-Sharp mass in spherical symmetry, equal to zero for every generic surface in AdS, has an appropriate form for gravity linearised about AdS and has an appropriate limit for large spheres in asymptotically AdS spacetimes.
    
\end{abstract}

\newpage
\tableofcontents

\noindent\rule{\textwidth}{0.8pt}

\section{Introduction}
One of the triumphs of mathematical general relativity is the positive energy theorem - originally proven by Schoen \& Yau \cite{Schoen1979} based on minimal surface techniques and soon after by Witten \cite{Witten1981} based on spinorial methods. Witten's method was subsequently extended to prove global mass-charge inequalities in 4D Einstein-Maxwell theory \cite{Gibbons1982}, 5D Einstein-Maxwell-Chern-Simons theory \cite{Gibbons1994}, global positive energy theorems for spacetimes with AdS-type asymptotics \cite{Wang2001, Chrusciel2001, Chrusciel2003, Cheng2005, Chrusciel2006, Wang2015} and some mass-charge inequalities in this context \cite{London1995, Kostelecky1996}.

Meanwhile, one of the outstanding problems in mathematical general relativity is to find a completely satisfactory definition of quasilocal mass, a notion of mass associated to a closed, compact, spacelike, 2D surface\footnote{More generally, this could be a co-dimension-2 surface, but this work is concerned with 4D spacetimes only.}, often taken to be diffeomorphic to a sphere - see \cite{Szabados2009} for a review on the many attempts in the literature. At a very high level, Witten's method equates a combination of the ADM quantities \cite{Arnowitt1962} to a non-negative volume integral over a Cauchy surface. This raises the tantalising possibility of replacing the Cauchy surface with a compact, spacelike, 3D hypersurface and thereby finding a notion of quasilocal mass on the hypersurface's boundary. Furthermore, such a quasilocal mass would likely automatically satisfy a form of positivity.

This idea culminated in the spinorial definition of quasilocal mass by Dougan \& Mason \cite{Dougan1991}, relying heavily on the Newman-Penrose (NP) \cite{Newman1962} and Geroch-Held-Penrose (GHP) \cite{Geroch1973} formalisms. Their definition proved to have a number of physically desirable properties \cite{Dougan1992} and simplified Penrose's \cite{Penrose1982} twistorial quasilocal mass\footnote{See also \cite{Zhang2009, Lott2023} for more recent spinorial definitions of quasilocal mass and \cite{Mondal2024} for an application of spinor methods to study positivity of quasilocal masses that are not themselves spinorial.}. The positive energy theorem associated with the Dougan-Mason mass was recently generalised to a quasilocal mass-charge inequality by Reall \cite{Reall2024}, in much the same way Gibbons \& Hull \cite{Gibbons1982} extended Witten's original work. In parallel with the increasing sophistication of global positive energy theorems, Reall suggests his results could be generalised to include a negative cosmological constant. However, to find a quasilocal mass-charge inequality - let alone apply it to the third law of black hole mechanics like Reall - one must first have a satisfactory notion of quasilocal mass for these spacetimes. While quasilocal masses do exist for spacetimes with negative cosmological constant - for example the Hawking mass \cite{Hawking1968} can be generalised \cite{Neves2010} and \cite{Wang2006} generalises the Brown-York and Kijowski masses \cite{Brown1993, Kijowski1997} - these are not naturally spinorial. Thus, one seeks a generalisation of the Dougan-Mason quasilocal mass accommodating a negative cosmological constant.

This work provides such a generalisation, stated as follows.
\begin{definition}[Quasilocal mass]
    \label{def:quasilocalMassIntroduction}
    Given a generic, 2D, spacelike surface, $S$, within a spacetime, $(M, g)$, satisfying the Einstein equation with negative cosmological constant, $\Lambda$, and matter fields satisfying the dominant energy condition, make the following constructions. Let $\{l, n, m, \overbar{m}\}$ be a NP tetrad adapted to $S$. Assume the null expansions of $S$ satisfy $\theta_l > 0$, $\theta_n < 0$ and $\theta_l\theta_n < \frac{2\Lambda}{3}$. Let $D_a$ be the Levi-Civita connection of $g$ and let
    \begin{align}
        \nabla_a\Psi = D_a\Psi + \mathrm{i}\sqrt{-\frac{\Lambda}{12}}\gamma_a\Psi
    \end{align}
    for any Dirac spinor, $\Psi$. Let $\Phi$ denote a Dirac spinor satisfying $\overbar{m}^a\nabla_a\Phi = 0$ on $S$ and let $\{\Phi^A\}$ be a basis of solutions. Then, define matrices, $Q^{AB}$ and $T^{AB}$, by
    \begin{align}
        Q^{AB} &= \int_Sl_an_b\left(\overline{\Phi}^A\gamma^{abc}\nabla_c\Phi^B - \overline{\nabla_c(\Phi^A)}\gamma^{abc}\Phi^B\right)\,\mathrm{d}A \\
        \mathrm{and}\,\,\,T^{AB} &= (\Phi^A)^TC^{-1}\Phi^B, \\
        \mathrm{where}\,\,\, C &= \mathrm{charge\,\,conjugation\,\,matrix}.
    \end{align}
    Finally, define the quasilocal mass as
    \begin{align}
        m(S) &= \frac{1}{16\pi}\sqrt{-\tr(QT^{-1}\overbar{Q}\overbar{T}^{-1})}.
    \end{align}
\end{definition}
The concept of ``generic" is discussed in greater detail later, but for the purpose of definition \ref{def:quasilocalMassIntroduction}, it amounts to $T^{-1}$ existing. Even then, a number of results are required to show this quasilocal mass is well-defined. However, one result is particularly crucial and non-trivial.
\begin{theorem}[Main theorem]
    If the dominant energy condition holds and the null expansions on $S$ satisfy $\theta_l > 0$, $\theta_n < 0$ \& $\theta_l\theta_n < \frac{2\Lambda}{3}$, then $Q^{AB}$ is a non-negative definite matrix.
\end{theorem}
Definition \ref{def:quasilocalMassIntroduction} is based both on Dougan \& Mason's work, but also on Penrose's twistorial definition\footnote{It appears there has been one previous attempt at including a negative cosmological constant in Penrose's work \cite{Kelly1985}. However, the definition in \cite{Kelly1985} is only evaluated at conformal infinity, $\mathcal{I}$. The present definition also differs in that no reference is made to twistors.}. Like the Dougan-Mason mass, this definition applies to 2D surfaces, $S$, which are ``generic" in a sense made precise later. Also, like Penrose's definition, but unlike Dougan \& Mason's, the present definition cannot decompose the mass into its constituents - e.g. energy and linear momentum - except near conformal infinity, $\mathcal{I}$. Most importantly though, a good quasilocal mass should satisfy several properties of physical significance. Although no unanimously agreed list exists, this paper shows the new quasilocal mass given satisfies the following physically desirable properties.
\begin{itemize}
    \item $m(S) \geq 0$.
    \item $m(S) = 0$ for every generic surface in AdS.
    \item $m(S)$ coincides with the Misner-Sharp mass \cite{Misner1964} (including cosmological constant) for spherically symmetric spacetimes.
    \item For asymptotically AdS spacetimes, $m(S)$ agrees with a global notion of mass as $S$ approaches a sphere on conformal infinity, $\mathcal{I}$.
    \item For gravity linearised about AdS, $m(S)$ agrees with a reasonable notion of mass built from the energy-momentum tensor, $T_{ab}$.
\end{itemize}
We begin in section \ref{sec:setup} by setting up the problem and establishing various foundational identities regarding spinors and the GHP formalism. This is supplemented in section \ref{sec:analysis} by analysis required to show a particular Dirac-type operator is invertible, as required for Witten's method. Finally, the new definition of quasilocal mass is stated in section \ref{sec:definition}; the first two properties in the list above are shown to follow somewhat immediately. Section \ref{sec:schwarzschild} is devoted to studying examples with high symmetry - namely spherical symmetry in section \ref{sec:spherical} and metrics with toroidal symmetry in section \ref{sec:torus}. Section \ref{sec:asymptotics} then establishes the asymptotic properties, while section \ref{sec:linearisation} studies gravity linearised around AdS. Section \ref{sec:conclusion} concludes the work with a recapitulation and some speculation on future research. Additionally, the conventions used are listed in appendix \ref{sec:conventions}; in short, they are based on \cite{Buchbinder1998}. These conventions differ slightly from alternative conventions used by Penrose and Rindler \cite{Penrose1984, Penrose1986} and a comparison is provided in appendix \ref{sec:penroseRindler}. Finally, appendix \ref{sec:identities} collates some identities frequently used when manipulating two-component spinors and NP coefficients.

\section{Set-up and the Lichnerowicz identity}
\label{sec:setup}
We begin with the basic set-up that will be used throughout. Let $\Sigma$ be a 3D, spacelike, compact submanifold with boundary, $S$, within a spacetime, $(M, g)$. Let $\{P, Q, X, Y\}$ be an orthonormal vierbein\footnote{Vierbeins are taken to be orthonormal by definition throughout this work.} with $X^a$ \& $Y^a$ tangent to $S$, $Q^a$ an outward-pointing normal to $S$ in $\Sigma$ and $P^a$ a timelike, future-directed normal to $\Sigma$. See figure \ref{fig:setup} for a visual depiction of this set-up. A quasilocal mass is then a number, $m(S)$, assigned to each applicable $S$. Note that $a, b \cdots$ will denote vierbein indices throughout, although in most equations they could equally well be replaced with abstract indices. Furthermore, the spacetime will always be assumed to solve the Einstein equation with a negative cosmological constant, namely
\begin{align}
    R_{ab} - \frac{1}{2}Rg_{ab} + \Lambda g_{ab} = 8\pi T_{ab},
\end{align}
where $R_{ab} - \frac{1}{2}Rg_{ab}$ is the Einstein tensor and $T_{ab}$ is the energy-momentum tensor.

Having chosen $\{P, Q, X, Y\}$ as described, define a Newman-Penrose (NP) tetrad \cite{Newman1962} by
    \begin{align}
        l^a &= \frac{1}{\sqrt{2}}\left(P^a + Q^a\right),\,\, n^a = \frac{1}{\sqrt{2}}\left(P^a - Q^a\right) \,\, \mathrm{and} \,\, m^a = \frac{1}{\sqrt{2}}\left(X^a + \mathrm{i}Y^a\right).
        \label{eq:NPTetrad}
    \end{align}
    Equivalently, given an NP tetrad adapted to $S$ and $\Sigma$, one can define
    \begin{align}
        &P^a = \frac{1}{\sqrt{2}}\left(l^a + n^a\right), \,\, Q^a = \frac{1}{\sqrt{2}}\left(l^a - n^a\right),\,\, X^a = \frac{1}{\sqrt{2}}\left(m^a + \overbar{m}^a\right) \nonumber \\
        &\mathrm{and} \,\, Y^a = \frac{1}{\mathrm{i}\sqrt{2}}\left(m^a - \overbar{m}^a\right).
    \end{align}
\begin{figure}
    \centering
    \includegraphics[scale=0.5]{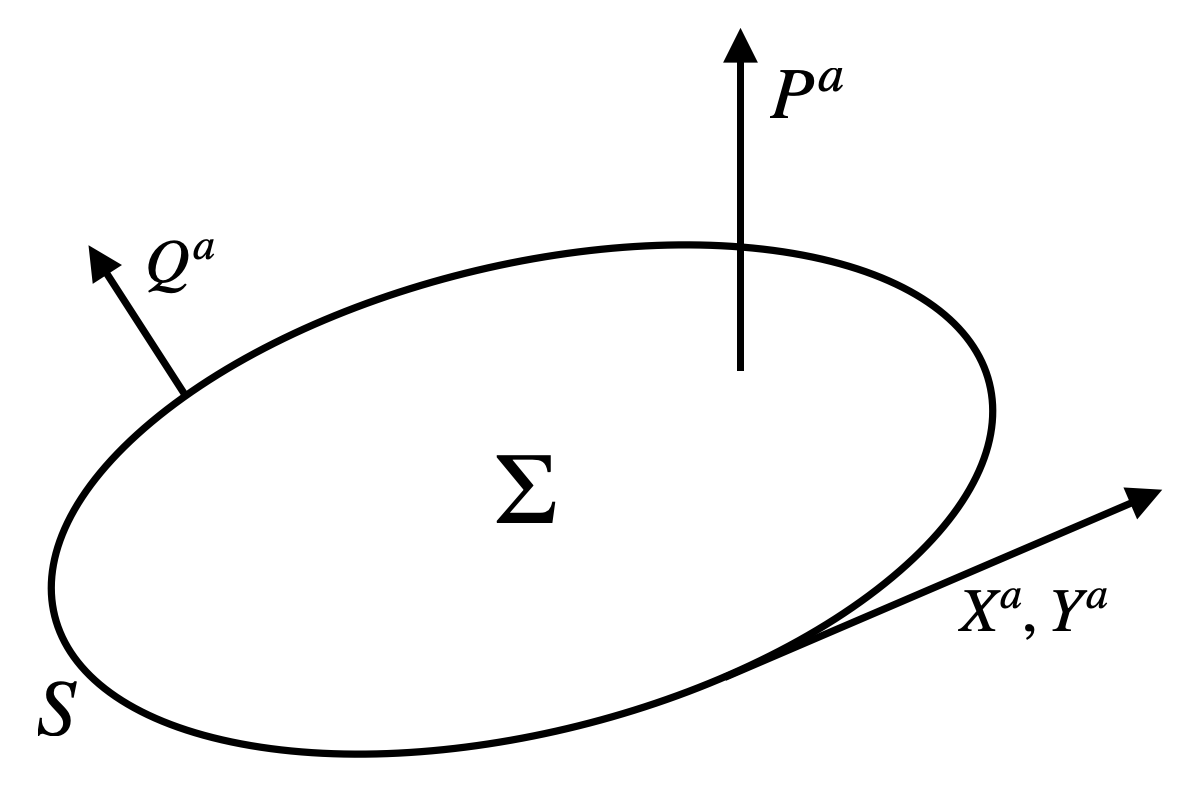}
    \caption{The set-up for defining quasilocal mass.}
    \label{fig:setup}
\end{figure}
\begin{lemma}
    The NP coefficients\footnote{See appendix \ref{sec:identities} for the definitions of the NP coefficients that appear in this work.}, $\mu$ and $\rho$, are real. Furthermore, they are related to the expansions along the null normals by $\theta_l = -2\rho$ and $\theta_n = 2\mu$.
    \label{thm:muRho}
\end{lemma}
\begin{proof}
    The reality of $\rho$ and $\mu$ for a surface forming tetrad is explained in \cite{Geroch1973}. It can be proven explicitly by using the symmetry of the extrinsic curvatures with respect to $l$ and $n$ respectively. For the expansions, if $\beta_{ab}$ is the induced metric on $S$, then
    \begin{align}
        \theta_l = \beta^{ab}D_al_b \,\,\,\mathrm{and}\,\,\,\theta_n = \beta^{ab}D_an_b. 
    \end{align}
    Since $g_{ab} = -l_an_b - n_al_b + m_a\overbar{m}_b + \overbar{m}_am_b$ in the NP formalism, it follows that $\beta_{ab} = m_a\overbar{m}_b + \overbar{m}_am_b$ and consequently $\theta_l = -\bar{\rho} - \rho = -2\rho$ and $\theta_n = \mu + \bar{\mu} = 2\mu$.
\end{proof}
Given a pair of null normals to $S$, it will be very natural to use the Geroch-Held-Penrose (GHP) formalism \cite{Geroch1973} in what follows. The primary construction underpinning the GHP formalism is the spinor dyad, $\{o, \iota\}$. In particular, when converted to two-component spinors, write the NP tetrad in terms of $\{o, \iota\}$ as 
    \begin{align}
        l_{\alpha\dot{\alpha}} = o_\alpha\bar{o}_{\dot{\alpha}} \,\,\,\mathrm{and}\,\,\, n_{\alpha\dot{\alpha}} = \iota_\alpha\bar{\iota}_{\dot{\alpha}}
        \label{eq:lnDyad}
    \end{align}
    with $\iota^\alpha o_\alpha = \sqrt{2}$. Subsequently, decompose any two-component spinor, $\psi_\alpha$, as 
    \begin{align}
        &\psi_\alpha = \psi_o o_\alpha + \psi_\iota\iota_\alpha
        \label{eq:psiOPsiI}\\
        &\iff \psi_o = \frac{1}{\sqrt{2}}\iota^\alpha\psi_\alpha \,\,\,\mathrm{and} \,\,\, \psi_\iota = -\frac{1}{\sqrt{2}}o^\alpha\psi_\alpha.
    \end{align}
    Finally, in terms of the spinor dyad, choose
    \begin{align}
        m_{\alpha\dot{\alpha}} = \iota_\alpha\bar{o}_{\dot{\alpha}} \,\,\,\mathrm{and}\,\,\, \overbar{m}_{\alpha\dot{\alpha}} = o_\alpha\bar{\iota}_{\dot{\alpha}}.
        \label{eq:mDyad}
    \end{align}
More detail on these constructions can be found in \cite{Geroch1973, Penrose1984, Penrose1986}, where $o_\alpha$ is $o_{A^\prime}$ and $\iota_\alpha$ is $\iota_{A^\prime}$. Also see appendix \ref{sec:penroseRindler} for a more comprehensive notation comparision to \cite{Geroch1973, Penrose1984, Penrose1986}.

\begin{definition}[Modified connection]
    When acting on any Dirac spinor, $\Psi$, of the spacetime, define the modified connection, $\nabla$, by
    \begin{align}
        \nabla_a\Psi &= D_a\Psi + \mathrm{i}k\gamma_a\Psi\,\,\,\,\,\,\,\mathrm{and} \\
        \nabla_a\overline{\Psi} &= D_a\overline{\Psi} - \mathrm{i}k\overline{\Psi}\gamma_a = (\nabla_a\Psi)^\dagger\gamma^0,
    \end{align}
    where $k = \sqrt{-\frac{\Lambda}{12}}$ is a positive constant, $D_a$ is the Levi-Civita connection of $g$ and $\{\gamma^a\}_{a = 0}^3$ are the Dirac gamma matrices\footnote{See appendix \ref{sec:conventions} for gamma matrix and Dirac spinor conjugation conventions.}.
\end{definition}
Note this connection is the transformation of the gravitino in a bosonic solution of gauged supergravity and hints at deeper links \cite{Horowitz1983} between supergravity and Witten's positive energy theorem proof.
\begin{definition}[$E^{ab}(\Psi), E^{ab}(\Psi_1, \Psi_2)$]
    For a Dirac spinor, $\Psi$, let 
    \begin{align}
        E^{ab}(\Psi) &= \overline{\Psi}\gamma^{abc}\nabla_c\Psi + \mathrm{c.c} = \overline{\Psi}\gamma^{abc}\nabla_c\Psi - \nabla_c(\overline{\Psi})\gamma^{abc}\Psi.
        \label{eq:eAB}
    \end{align}
    Similarly, define $E^{ab}(\Psi_1, \Psi_2)$ by 
    \begin{align}
        E^{ab}(\Psi_1, \Psi_2) &= \overline{\Psi}_1\gamma^{abc}\nabla_c\Psi_2 - \nabla_c(\overline{\Psi}_1)\gamma^{abc}\Psi_2.
    \end{align}
\end{definition}
$E^{ab}(\Psi)$ is the Hodge dual of what is usually called the Witten-Nester 2-form \cite{Nester1981}.
\begin{theorem}[Lichnerowicz identity]
    \label{thm:lichnerowicz}
    Let $I, J, \cdots$ denote vierbein indices, but running over only the 1, 2 and 3 components, i.e. $\{Q, X, Y\}$ components. Then, the Lichnerowicz identity can be written as
    \begin{align}
        P_aD_b(E^{ba}(\Psi)) &= 2\left(\nabla_I(\Psi)^\dagger\nabla^I\Psi - 4\pi T^{0a}\overline{\Psi}\gamma_a\Psi - (\gamma^I\nabla_I\Psi)^\dagger\gamma^J\nabla_J\Psi\right).
    \end{align}
\end{theorem}
A variant of the Lichnerowicz identity is always the key result underpinning any Witten-style positive energy theorem. Note the RHS can be written in a more covariant-looking way by replacing $T^{0a}$ with $-P_bT^{ba}$ and replacing all $\nabla_I$ with $h\indices{^b_a}\nabla_b$, where $h_{ab} = g_{ab} + P_aP_b$.
\begin{proof}
    By directly expanding and recombining terms, one finds
    \begin{align}
        D_bE^{ba}(\Psi) &= 2\nabla_b(\overline{\Psi})\gamma^{bac}\nabla_c\Psi - 12k^2\overline{\Psi}\gamma^{a}\Psi - \overline{\Psi}\gamma^{abc}D_b D_c\Psi - D_bD_c(\overline{\Psi})\gamma^{cba}\Psi.
    \end{align}
    For the second derivative terms, recall the standard Lichnerowicz identity, 
    \begin{align}
        \gamma^{abc}D_bD_c\Psi &= -\frac{1}{2}\left(R^{ab} - \frac{1}{2}g^{ab}R\right)\gamma_b\Psi,
        \label{eq:lichnerowicz}
    \end{align}
    which is proven using $\gamma^{abc}$'s antisymmetry to convert $D_bD_c\Psi$ into $\frac{1}{2}[D_b, D_c]\Psi = -\frac{1}{8}R\indices{^d^e_b_c}\gamma_{de}\Psi$ and then applying the identity, $\gamma^{abc}\gamma_{de} = \gamma\indices{^a^b^c_d_e} - 6\gamma\indices{^[^a^b_[_e}\delta\indices{^c^]_d_]} + 6\gamma^{[a}\delta\indices{^b_[_e}\delta\indices{^c^]_d_]}$. Consequently, for the other second derivative term,
    \begin{align}
        D_bD_c(\overline{\Psi})\gamma^{cba} &= (\gamma^{abc}D_b D_c\Psi)^\dagger\gamma^0 = -\frac{1}{2}\left(R^{ab} - \frac{1}{2}g^{ab}R\right)\overline{\Psi}\gamma_b.
    \end{align}
    Substituting back,
    \begin{align}
        D_bE^{ba}(\Psi) &= 2\nabla_b(\overline{\Psi})\gamma^{bac}\nabla_c\Psi - 12k^2\overline{\Psi}\gamma^{a}\Psi + \left(R^{ab} - \frac{1}{2}g^{ab}R\right)\overline{\Psi}\gamma_b\Psi \\
        &= 2\nabla_b(\overline{\Psi})\gamma^{bac}\nabla_c\Psi + 8\pi T^{ab}\overline{\Psi}\gamma_b\Psi
    \end{align}
    by the Einstein equation. $P_a \equiv -\delta_{a0}$ in vierbein indices. Hence,
    \begin{align}
        P_aD_bE^{ba}(\Psi)
        &= 2\left(\nabla_I(\Psi)^\dagger\gamma^{IJ}\nabla_J\Psi - 4\pi T^{0a}\overline{\Psi}\gamma_a\Psi\right) \\
        &= 2\left(\nabla_I(\Psi)^\dagger\nabla^I\Psi - (\gamma^I\nabla_I\Psi)^\dagger\gamma^J\nabla_J\Psi - 4\pi T^{0a}\overline{\Psi}\gamma_a\Psi\right),
    \end{align}
    which is the form of the Lichnerowicz identity needed in this work.
\end{proof}
It will prove useful to apply the Lichnerowicz identity together with the following result.
\begin{lemma}
    For any spacetime antisymmetric tensor, $M^{ab}$, 
    \begin{align}
        P_aD_bM^{ba} &= \widetilde{D}_b(P_aM^{ba}),
    \end{align}
    where $\widetilde{D}$ is the induced covariant derivative on $\Sigma$.
    \label{thm:antisymmetricDerivative}
\end{lemma}
\begin{proof}
    This result is from \cite{Cheng2005}. It can be proven as follows. Let $h_{ab}$ be the induced metric on $\Sigma$, i.e. $h_{ab} = g_{ab} + P_aP_b$. Then, since $h\indices{^a_c}P_bM^{cb} = P_bM^{ab}$ by $M^{ab}$'s antisymmetry, the induced covariant derivative acts as
    \begin{align}
        \widetilde{D}_b(P_aM^{ba}) &= h\indices{^c_b}D_{c}(P_aM^{ba}) \\
        &= K_{ba}M^{ba} + \delta\indices{^c_b}P_aD_{c}M^{ba} + P^cP_bP_aD_{c}M^{ba} \\
        &= P_aD_bM^{ba},
    \end{align}
    where $K_{ab}$ is the extrinsic curvature of $\Sigma$ in $M$ and $M^{ba}$'s antisymmetry has been applied.
\end{proof}
\begin{definition}[$Q(\Psi), Q(\Psi_1, \Psi_2)$]
    \label{def:q}
    For a Dirac spinor, $\Psi$, define $Q(\Psi)$ by 
    \begin{align}
        Q(\Psi) &= \int_{\Sigma}P_aD_b(E^{ba}(\Psi))\mathrm{d}V.
        \label{eq:q}
    \end{align}
    By lemma \ref{thm:antisymmetricDerivative}, 
    \begin{align}
        Q(\Psi) &= \int_{S}P_aQ_bE^{ba}(\Psi)\mathrm{d}A. 
    \end{align}
    Meanwhile, by theorem \ref{thm:lichnerowicz},
    \begin{align}
        Q(\Psi) &= 2\int_{\Sigma}\left(\nabla_I(\Psi)^\dagger\nabla^I\Psi - 4\pi T^{0a}\overline{\Psi}\gamma_a\Psi - (\gamma^I\nabla_I\Psi)^\dagger\gamma^J\nabla_J\Psi\right)\mathrm{d}V.
    \end{align}
    Similarly, define $Q(\Psi_1, \Psi_2)$ by 
    \begin{align}
        Q(\Psi_1, \Psi_2) &= \int_{\Sigma}P_aD_b(E^{ba}(\Psi_1, \Psi_2))\mathrm{d}V.
    \end{align}
\end{definition}
Although Dirac spinors are more convenient on $\Sigma$, the positive energy theorem associated to the Dougan-Mason construction requires two-component spinors on $S$.
\begin{lemma}
    If $\Psi = (\psi_\alpha, \bar{\chi}^{\dot{\alpha}})^T$, then
    \begin{align}
        Q(\Psi) &= 4\int_S\Big(\psi_\iota\eth\overline{\psi}_o + \overline{\psi}_\iota\bar{\eth}\psi_o - \overline{\chi}_o\eth\chi_\iota - \chi_o\bar{\eth}\overline{\chi}_\iota + \rho|\psi_o|^2 + \mu|\psi_\iota|^2 + \rho|\chi_o|^2 + \mu|\chi_\iota|^2 \nonumber \\
        &\,\,\,\,\,\,\,\,\,\,\,\,\,\,\,\,\,\, + \mathrm{i}k\sqrt{2}\big(\psi_o\chi_\iota + \psi_\iota\chi_o - \overline{\psi}_o\overline{\chi}_\iota - \overline{\psi}_\iota\overline{\chi}_o\big)\Big)\mathrm{d}A,
    \end{align}
    where $\eth$ and $\bar{\eth}$ are the eth and eth-bar operators defined in \cite{Geroch1973}.
    \label{thm:qGHP}
\end{lemma}
\begin{proof}
    By equation \ref{eq:NPTetrad}, definition \ref{def:q} and $E^{ba}$'s antisymmetry, 
    \begin{align}
        Q(\Psi) &= \int_{S}P_aQ_bE^{ba}(\Psi)\mathrm{d}A = \int_{S}l_an_bE^{ab}(\Psi)\mathrm{d}A = \frac{1}{4}\int_{S}l_{\alpha\dot{\alpha}}n_{\beta\dot{\beta}}E^{\alpha\dot{\alpha}\beta\dot{\beta}}(\Psi)\mathrm{d}A.
        \label{eq:qLNE}
    \end{align}
    Finding $E_{\alpha\dot{\alpha}\beta\dot{\beta}}(\Psi)$ is a long, tedious calculation, summarised as follows.
    \begin{align}
        E_{\alpha\dot{\alpha}\beta\dot{\beta}}(\Psi) &= (\sigma_a)_{\alpha\dot{\alpha}}(\sigma_b)_{\beta\dot{\beta}}E^{ab}(\Psi) = (\sigma_a)_{\alpha\dot{\alpha}}(\sigma_b)_{\beta\dot{\beta}}(\overline{\Psi}\gamma^{abc}\nabla_c\Psi - \nabla_c(\overline{\Psi})\gamma^{abc}\Psi).
    \end{align}
    It suffices to study only the first term; the second term is just the complex conjugate. By applying the Weyl representation of the gamma matrices, as given in appendix \ref{sec:conventions}, one finds
    \begin{align}
        &(\sigma_a)_{\alpha\dot{\alpha}}(\sigma_b)_{\beta\dot{\beta}}\overline{\Psi}\gamma^{abc}\nabla_c\Psi \nonumber \\
        &= -(\sigma_a)_{\alpha\dot{\alpha}}(\sigma_b)_{\beta\dot{\beta}}\chi^\gamma(\sigma^{[a}\tilde{\sigma}^b\sigma^{c]})_{\gamma\dot{\gamma}}D_c\bar{\chi}^{\dot{\gamma}} - (\sigma_a)_{\alpha\dot{\alpha}}(\sigma_b)_{\beta\dot{\beta}}\overline{\psi}_{\dot{\gamma}}(\tilde{\sigma}^{[a}\sigma^b\tilde{\sigma}^{c]})^{\dot{\gamma}\gamma}D_c\psi_\gamma \nonumber \\
        &\,\,\,\,\,\,\, + 2\mathrm{i}k(\sigma_a)_{\alpha\dot{\alpha}}(\sigma_b)_{\beta\dot{\beta}}\chi^\gamma(\sigma^{[a}\tilde{\sigma}^{b]})\indices{_\gamma^\delta}\psi_\delta + 2\mathrm{i}k(\sigma_a)_{\alpha\dot{\alpha}}(\sigma_b)_{\beta\dot{\beta}}\overline{\psi}_{\dot{\gamma}}(\tilde{\sigma}^{[a}\sigma^{b]})\indices{^{\dot{\gamma}}_{\dot{\delta}}}\bar{\chi}^{\dot{\delta}}.
        \label{eq:randoEquation0}
    \end{align}
    Consider this expression term by term. From the identity, 
    \begin{align}
        (\sigma_a)_{\alpha\dot{\beta}}(\tilde{\sigma}_b)^{\dot{\beta}\beta}(\sigma_c)_{\beta\dot{\alpha}} &= g_{ca}(\sigma_b)_{\alpha\dot{\alpha}} - g_{bc}(\sigma_a)_{\alpha\dot{\alpha}} - g_{ab}(\sigma_c)_{\alpha\dot{\alpha}} + \mathrm{i}\varepsilon_{abcd}(\sigma^d)_{\alpha\dot{\alpha}},
    \end{align}
    it follows that
    \begin{align}
    (\sigma_{[a}\tilde{\sigma}_b\sigma_{c]})_{\alpha\dot{\alpha}} &= \mathrm{i}\varepsilon_{abcd}(\sigma^d)_{\alpha\dot{\alpha}} = (\sigma_a)_{\alpha\dot{\beta}}(\tilde{\sigma}_b)^{\dot{\beta}\beta}(\sigma_c)_{\beta\dot{\alpha}} - g_{ca}(\sigma_b)_{\alpha\dot{\alpha}} + g_{bc}(\sigma_a)_{\alpha\dot{\alpha}} + g_{ab}(\sigma_c)_{\alpha\dot{\alpha}}.
    \end{align}
    In conjunction with the identities in appendix \ref{sec:identities}, the original term then reduces to
    \begin{align}
        -(\sigma_a)_{\alpha\dot{\alpha}}(\sigma_b)_{\beta\dot{\beta}}\chi^\gamma(\sigma^{[a}\tilde{\sigma}^b\sigma^{c]})_{\gamma\dot{\gamma}}D_c\bar{\chi}^{\dot{\gamma}} 
        &= 2\chi_\alpha D_{\beta\dot{\alpha}}\bar{\chi}_{\dot{\beta}} - 2\chi_\beta D_{\alpha\dot{\beta}}\bar{\chi}_{\dot{\alpha}}.
    \end{align}
    The other terms in equation \ref{eq:randoEquation0} are handled similarly to get
    \begin{align}
        &-(\sigma_a)_{\alpha\dot{\alpha}}(\sigma_b)_{\beta\dot{\beta}}\overline{\psi}_{\dot{\gamma}}(\tilde{\sigma}^{[a}\sigma^b\tilde{\sigma}^{c]})^{\dot{\gamma}\gamma}D_c\psi_\gamma = 2\overline{\psi}_{\dot{\alpha}}D_{\alpha\dot{\beta}}\psi_\beta - 2\overline{\psi}_{\dot{\beta}}D_{\beta\dot{\alpha}}\psi_\alpha \,\,\,\mathrm{and} \\
        &2\mathrm{i}k(\sigma_a)_{\alpha\dot{\alpha}}(\sigma_b)_{\beta\dot{\beta}}\chi^\gamma(\sigma^{[a}\tilde{\sigma}^{b]})\indices{_\gamma^\delta}\psi_\delta + 2\mathrm{i}k(\sigma_a)_{\alpha\dot{\alpha}}(\sigma_b)_{\beta\dot{\beta}}\overline{\psi}_{\dot{\gamma}}(\tilde{\sigma}^{[a}\sigma^{b]})\indices{^{\dot{\gamma}}_{\dot{\delta}}}\bar{\chi}^{\dot{\delta}} \nonumber \\
        &= -4\mathrm{i}k\varepsilon_{\alpha
        \beta}(\overline{\psi}_{\dot{\alpha}}\bar{\chi}_{\dot{\beta}} + \overline{\psi}_{\dot{\beta}}\bar{\chi}_{\dot{\alpha}}) + 4\mathrm{i}k\varepsilon_{\dot{\alpha}\dot{\beta}}(\chi_\alpha\psi_\beta + \chi_\beta\psi_\alpha).
    \end{align}
    Putting it all together, equation \ref{eq:randoEquation0} reduces to
    \begin{align}
        (\sigma_a)_{\alpha\dot{\alpha}}(\sigma_b)_{\beta\dot{\beta}}\overline{\Psi}\gamma^{abc}\nabla_c\Psi
        &= 2\chi_\alpha D_{\beta\dot{\alpha}}\bar{\chi}_{\dot{\beta}} - 2\chi_\beta D_{\alpha\dot{\beta}}\bar{\chi}_{\dot{\alpha}} + 2\overline{\psi}_{\dot{\alpha}}D_{\alpha\dot{\beta}}\psi_\beta - 2\overline{\psi}_{\dot{\beta}}D_{\beta\dot{\alpha}}\psi_\alpha \nonumber \\
        &\,\,\,\,\,\,\, - 4\mathrm{i}k\varepsilon_{\alpha
        \beta}(\overline{\psi}_{\dot{\alpha}}\bar{\chi}_{\dot{\beta}} + \overline{\psi}_{\dot{\beta}}\bar{\chi}_{\dot{\alpha}}) + 4\mathrm{i}k\varepsilon_{\dot{\alpha}\dot{\beta}}(\chi_\alpha\psi_\beta + \chi_\beta\psi_\alpha).
    \end{align}
    Then, by adding the complex conjugate,
    \begin{align}
        E_{\alpha\dot{\alpha}\beta\dot{\beta}} &= 2\big(\chi_\alpha D_{\beta\dot{\alpha}}\bar{\chi}_{\dot{\beta}} - \chi_\beta D_{\alpha\dot{\beta}}\bar{\chi}_{\dot{\alpha}} + \overline{\psi}_{\dot{\alpha}}D_{\alpha\dot{\beta}}\psi_\beta - \overline{\psi}_{\dot{\beta}}D_{\beta\dot{\alpha}}\psi_\alpha + \bar{\chi}_{\dot{\alpha}}D_{\alpha\dot{\beta}}\chi_\beta - \bar{\chi}_{\dot{\beta}}D_{\beta\dot{\alpha}}\chi_\alpha \nonumber \\
        &\,\,\,\,\,\,\, + \psi_\alpha D_{\beta\dot{\alpha}}\overline{\psi}_{\dot{\beta}}  - \psi_\beta D_{\alpha\dot{\beta}}\overline{\psi}_{\dot{\alpha}}\big) + 8\mathrm{i}k\big(-\varepsilon_{\alpha
        \beta}(\overline{\psi}_{\dot{\alpha}}\bar{\chi}_{\dot{\beta}} + \overline{\psi}_{\dot{\beta}}\bar{\chi}_{\dot{\alpha}}) + \varepsilon_{\dot{\alpha}\dot{\beta}}(\psi_\alpha\chi_\beta + \psi_\beta\chi_\alpha)\big).
    \end{align}
    Therefore, by equations \ref{eq:lnDyad} and \ref{eq:mDyad}, the required integrand is 
    \begin{align}
        l^{\alpha\dot{\alpha}}n^{\beta\dot{\beta}}E_{\alpha\dot{\alpha}\beta\dot{\beta}}(\Psi) &= o^{\alpha}\bar{o}^{\dot{\alpha}}\iota^\beta\bar{\iota}^{\dot{\beta}}E_{\alpha\dot{\alpha}\beta\dot{\beta}} \\
        &= 4\sqrt{2}\big(\chi_\iota\bar{\iota}^{\dot{\beta}}\delta\bar{\chi}_{\dot{\beta}} + \chi_o\bar{o}^{\dot{\alpha}}\bar{\delta}\bar{\chi}_{\dot{\alpha}} + \overline{\psi}_\iota\iota^\beta\bar{\delta}\psi_\beta + \overline{\psi}_o o^\alpha\delta\psi_\alpha + \overline{\chi}_\iota\iota^\beta\bar{\delta}\chi_\beta + \overline{\chi}_oo^\alpha\delta\chi_\alpha \nonumber \\
        &\,\,\,\,\,\,\, + \psi_\iota\bar{\iota}^{\dot{\beta}}\delta\overline{\psi}_{\dot{\beta}} + \psi_o\bar{o}^{\dot{\alpha}}\bar{\delta}\,\overline{\psi}_{\dot{\alpha}}\big) + 16\sqrt{2}\mathrm{i}k\big(-\overline{\psi}_\iota\overline{\chi}_o - \overline{\psi}_o\overline{\chi}_\iota + \psi_\iota\chi_o + \psi_o\chi_\iota\big),
    \end{align}
    where $\delta$ and $\bar{\delta}$ denote $m^aD_a$ and $\overbar{m}^aD_a$ respectively. All the derivative terms can be re-written in terms of the GHP $\eth$ and $\bar{\eth}$ operators \cite{Geroch1973}. In particular, $o_\alpha$ and $\iota_\alpha$ are GHP type-(0, 1) and type-(0, -1) respectively by definition\footnote{In the GHP formalism, an object, $f_{p, q}$, is said to be of type-$(p, q)$ if and only if $f_{p, q} \to z^p\bar{z}^qf_{p, q}$ when $\bar{o}^{\dot{\alpha}} \to z\bar{o}^{\dot{\alpha}}$ and $\bar{\iota}^{\dot{\alpha}} \to \bar{\iota}^{\dot{\alpha}}/z$.}. Since $\psi_\alpha$ (and likewise for $\chi_\alpha$ etc.) is invariant under choice of spinor dyad, it must be that $\psi_o$ and $\psi_\iota$ are type-(0, -1) and type-(0, 1) respectively. For a type-$(p, q)$ object, $f_{p, q}$, in terms of NP quantities, $\eth$ and $\bar{\eth}$ are defined to act as
    \begin{align}
        \eth f_{p, q} &= \delta f_{p, q} - p\beta f_{p, q} - q\bar{\alpha}f_{p, q}
        \label{eq:edth} \\
        \mathrm{and}\,\,\bar{\eth}f_{p, q} &= \bar{\delta}f_{p, q} - p\alpha f_{p, q} - q\bar{\beta}f_{p, q}.
        \label{eq:edthBar}
    \end{align}
    It immediately follows that
    \begin{align}
        \bar{\iota}^{\dot{\beta}}\delta\bar{\chi}_{\dot{\beta}} &= \sqrt{2}\left(\eth\overline{\chi}_o + \mu\overline{\chi}_\iota\right), \\
        \bar{o}^{\dot{\alpha}}\bar{\delta}\bar{\chi}_{\dot{\alpha}} &= -\sqrt{2}\left(\bar{\eth}\overline{\chi}_\iota - \rho\overline{\chi}_o\right), \\
        \iota^\beta\bar{\delta}\psi_\beta &= \sqrt{2}\left(\bar{\eth}\psi_o + \mu\psi_\iota\right) \,\,\,\mathrm{and} \\
        o^\alpha\delta\psi_\alpha &= -\sqrt{2}\left(\eth \psi_\iota - \rho\psi_o\right).
    \end{align}
    Substituting back, simplifying and taking complex conjugates as needed yields
    \begin{align}
        l^{\alpha\dot{\alpha}}n^{\beta\dot{\beta}}E_{\alpha\dot{\alpha}\beta\dot{\beta}}(\Psi) &= 8\big(\chi_\iota\eth\overline{\chi}_o - \chi_o\bar{\eth}\overline{\chi}_\iota + \overline{\psi}_\iota\bar{\eth}\psi_o - \overline{\psi}_o\eth\psi_\iota + \overline{\chi}_\iota\bar{\eth}\chi_o - \overline{\chi}_o\eth\chi_\iota + \psi_\iota\eth\overline{\psi}_o - \psi_o\bar{\eth}\overline{\psi}_\iota\big) \nonumber \\
        &\,\,\,\,\,\,\, + 16\big(\mu|\chi_\iota|^2 + \rho|\chi_o|^2 + \mu|\psi_\iota|^2 + \rho|\psi_o|^2\big) \nonumber \\
        &\,\,\,\,\,\,\, + 16\sqrt{2}\mathrm{i}k\big(-\overline{\psi}_\iota\overline{\chi}_o - \overline{\psi}_o\overline{\chi}_\iota + \psi_\iota\chi_o + \psi_o\chi_\iota\big).
        \label{eq:randoEquation1}
    \end{align}
    The $\eth$ and $\bar{\eth}$ operators were constructed by GHP \cite{Geroch1973} such that integration by parts is valid on $S$, e.g.
    \begin{align}
        \int_S\overline{\psi}_o\eth(\psi_\iota)\mathrm{d}A = -\int_S\psi_\iota\eth(\overline{\psi}_o)\mathrm{d}A.
    \end{align}
    Substituting equation \ref{eq:randoEquation1} into equation \ref{eq:qLNE} and integrating by parts proves the lemma.
\end{proof}

\section{Elements of analysis}
\label{sec:analysis}
A key idea of Witten's method is applying the Lichnerowicz identity with a spinor, $\Psi$, solving $\gamma^I\nabla_I\Psi = 0$ on $\Sigma$. This section is dedicated to proving this is always possible given appropriate boundary conditions on $S$ and given an appropriate functional space for $\Psi$. Readers willing to take this fact for granted can skip ahead to section \ref{sec:definition}. The presentation here is heavily based on \cite{Bartnik2005, ChruscielBartnik2003, Chrusciel2010}. 
\begin{definition}[$C_b^\infty$]
    Let $C_b^\infty$ be the space of Dirac spinors, $\Psi = (\psi_\alpha, \overbar{\chi}^{\dot{\alpha}})^T$, which are smooth on $\Sigma$ and subject to the boundary conditions, $\psi_o = \chi_\iota = 0$ on $S$.
    \label{def:cbInfinity}
\end{definition}
\begin{lemma}[$\langle\cdot, \cdot\rangle_{C_b^\infty}$]
    \label{def:cbInnerProduct}
    Assume the dominant energy condition holds on $\Sigma$ and the null expansions on $S$ satisfy $\theta_l > 0$, $\theta_n < 0$ \& $\theta_l\theta_n < -8k^2$. Then, 
    \begin{align}
        \langle\Psi_1, \Psi_2\rangle_{C_b^\infty} &= \int_\Sigma\left((\nabla_I\Psi_1)^\dagger\nabla^I\Psi_2 + 4\pi T^{0a}\Psi_1^\dagger\gamma_0\gamma_a\Psi_2\right)\mathrm{d}V - Q(\Psi_1, \Psi_2)
    \end{align}
    defines an inner product on $C_b^\infty$.
\end{lemma}
\begin{proof}
    Conjugate symmetry and linearity in the second argument are manifest, so only positive definiteness needs to be checked.
    \begin{align}
        \langle\Psi, \Psi\rangle_{C_b^\infty} &= \int_\Sigma\left((\nabla_I\Psi)^\dagger\nabla^I\Psi + 4\pi T^{0a}\Psi^\dagger\gamma_0\gamma_a\Psi\right)\mathrm{d}V - Q(\Psi).
        \label{eq:cbNorm}
    \end{align}
    $\Psi \in C_b^\infty$ implies $\psi_o = \chi_\iota = 0$ on $S$. Then, by lemma \ref{thm:qGHP},
    \begin{align}
        Q(\Psi) &= 4\int_S\Big(\mu|\psi_\iota|^2 + \rho|\chi_o|^2 + \mathrm{i}k\sqrt{2}\big(\psi_\iota\chi_o - \overline{\psi}_\iota\overline{\chi}_o\big)\Big)\mathrm{d}A.
    \end{align}
    For any nonwhere vanishing, complex function, $z$, on $S$, $Q(\Psi)$ can also be written as\footnote{The next equation is formally identical to GHP boost invariance. I'd like to thank Harvey Reall for pointing out this process allows for a single inequality, $\theta_l\theta_n < -8k^2$, as opposed to a pair of inequalities, $\theta_l > 2\sqrt{2}k$ and $\theta_n < -2\sqrt{2}k$.}
    \begin{align}
        Q(\Psi) &= 4\int_S\bigg(\frac{\mu}{|z|^2}|z\psi_\iota|^2 + \rho|z|^2\bigg|\frac{\chi_o}{z}\bigg|^2 + \mathrm{i}k\sqrt{2}\bigg(z\psi_\iota\frac{\chi_o}{z} - \bar{z}\overline{\psi}_\iota\frac{\overline{\chi}_o}{\bar{z}}\bigg)\bigg)\mathrm{d}A
    \end{align}
    Let $\mu^\prime = \mu/|z|^2$, $\rho^\prime = |z|^2\rho$, $\psi_\iota^\prime = z\psi_\iota$ and $\chi_o^\prime = \chi_o/z$. Then,
    \begin{align}
        Q(\Psi) &= 4\int_S\Big(\mu^\prime|\psi^\prime_\iota|^2 + \rho^\prime|\chi^\prime_o|^2 + \mathrm{i}k\sqrt{2}\big(\psi^\prime_\iota\chi^\prime_o - \overline{\psi}^\prime_\iota\overline{\chi}^\prime_o\big)\Big)\mathrm{d}A \\
        &= 4\int_S\Big((\mu^\prime + k\sqrt{2})|\psi^\prime_\iota|^2 + (\rho^\prime + k\sqrt{2})|\chi^\prime_o|^2  - k\sqrt{2}|\psi^\prime_\iota + \mathrm{i}\overline{\chi}^\prime_o|^2\Big)\mathrm{d}A \\
        &\leq 4\int_S\Big((\mu^\prime + k\sqrt{2})|\psi^\prime_\iota|^2 + (\rho^\prime + k\sqrt{2})|\chi^\prime_o|^2\Big)\mathrm{d}A.
        \label{eq:qPsiCb}
    \end{align}
    Choose $z = \sqrt[4]{\mu/\rho}$ so that $\mu^\prime = \rho^\prime = -\sqrt{\mu\rho} = -\frac{1}{2}\sqrt{-\theta_l\theta_n} < -k\sqrt{2}$ by lemma \ref{thm:muRho}. Therefore it immediately follows that $Q(\Psi) \leq 0$.

    Next, consider $T^{0a}\Psi^\dagger\gamma_0\gamma_a\Psi = \Psi^\dagger(T^{00}I + T^{0I}\gamma_0\gamma_I)\Psi$. The eigenvalues of $T^{0I}\gamma_0\gamma_I$ are $\pm\sqrt{T^{0I}T\indices{^0_I}}$, so $T^{0a}\gamma_0\gamma_a$ is non-negative definite if and only if $T^{00} \geq \sqrt{T^{0I}T\indices{^0_I}}$, which holds from the dominant energy condition. In summary, all three terms in equation \ref{eq:cbNorm} are non-negative and therefore
    $\langle\Psi, \Psi\rangle_{C_b^\infty} \geq 0$.

    Finally, suppose $\langle\Psi, \Psi\rangle_{C_b^\infty} = 0$. Then, by equation \ref{eq:cbNorm}, $\nabla_I\Psi = 0$ on $\Sigma$ and $Q(\Psi) = 0$. The boundary conditions already imply $\psi_o = \chi_\iota = 0$ on $S$, equation \ref{eq:qPsiCb} then implies $\chi_o = \psi_\iota = 0$ on $S$ too and therefore $\Psi = 0$ on $S$. Choose an arbitrary point, $p \in \Sigma$, and a smooth curve from any point on $S$ to $p$. Let $t^I$ be the tangent to the curve. Then, $t^I\nabla_I\Psi = 0$ along the curve and $\Psi = 0$ at the initial point. It follows that $\Psi = 0$ everywhere along the curve since $\Psi$ is smooth and thus 1st order, linear, homogeneous ODEs will have a unique solution. But, $p$ is arbitrary, so $\Psi = 0$ everywhere on $\Sigma$. Therefore, one concludes $\langle\cdot, \cdot\rangle_{C_b^\infty}$ is positive definite.
\end{proof}
\begin{definition}[$\mathfrak{D}$]
    Define a linear operator, $\mathfrak{D} : C_b^\infty \to L^2$, by $\mathfrak{D} : \Psi \mapsto \gamma^I\nabla_I\Psi$.
    \label{def:G}
\end{definition}
\begin{lemma}
    $\langle\Psi_1, \Psi_2\rangle_{C_b^\infty} = \langle \mathfrak{D}(\Psi_1), \mathfrak{D}(\Psi_2)\rangle_{L^2}$.
    \label{thm:cbToL2}
\end{lemma}
\begin{proof}
    Apply theorem \ref{thm:lichnerowicz}, definition \ref{def:q} and the polarisation identity for relating norms and inner products.
\end{proof}
\begin{definition}[$\mathcal{H}$]
    Define $\mathcal{H}$ to be the completion\footnote{While the space defined by this completion appears fairly abstract, it can be seen from the methods in section three of \cite{Bartnik2005, ChruscielBartnik2003} that because $\Sigma$ is compact, a Poincar\'{e} inequality holds, meaning elements of $\mathcal{H}$ can be viewed as elements of the familiar Sobolev space, $H^1$.} of $C_b^\infty$ under $\langle\cdot, \cdot\rangle_{C_b^\infty}$.
    \label{def:h}
\end{definition}
\begin{lemma}
    $\mathfrak{D}$ extends to a continuous (i.e. bounded) linear operator from $\mathcal{H}$ to $L^2$ such that $\langle\Psi_1, \Psi_2\rangle_{\mathcal{H}} = \langle \mathfrak{D}(\Psi_1), \mathfrak{D}(\Psi_2)\rangle_{L^2}$.
    \label{thm:GExtension}
\end{lemma}
\begin{proof}
    The points in $\mathcal{H}\backslash C_b^\infty$ are equivalence classes of Cauchy sequences. Let $\{\Psi_A\}_{A = 0}^\infty$ be one such Cauchy sequence in $C_b^\infty$ with limit in $\mathcal{H}\backslash C_b^\infty$. Observe that by lemma \ref{thm:cbToL2},
    \begin{align}
        ||\mathfrak{D}(\Psi_A) - \mathfrak{D}(\Psi_B)||_{L^2} = ||\mathfrak{D}(\Psi_A - \Psi_B)||_{L^2} = ||\Psi_A - \Psi_B||_{C_b^\infty}.
    \end{align}
    Thus $\{\mathfrak{D}(\Psi_A)\}_{A = 0}^\infty$ is a Cauchy sequence in $L^2$ and since $L^2$ is complete, $\exists \lim_{A\to\infty}\mathfrak{D}(\Psi_A) \in L^2$. Extend the definition of $\mathfrak{D}$ to $\mathcal{H}\backslash C_b^\infty$ by defining\footnote{This definition is independent of the original choice of Cauchy sequence, $\{\Psi_A\}_{A = 0}^\infty$, because choosing a different Cauchy sequence with the same ``limit," $\{\Psi^\prime_A\}_{A = 0}^\infty$, implies $\{\mathfrak{D}(\Psi_A), \mathfrak{D}(\Psi^\prime_B)\}$ is a Cauchy sequence in $L^2$ by a similar computation to above. Hence, they would have the same limit in $L^2$.} $\mathfrak{D}(\lim_{A\to\infty}\Psi_A) = \lim_{A\to\infty}\mathfrak{D}(\Psi_A)$.

    Next, observe that this definition implies lemma \ref{thm:cbToL2} extends to $\mathcal{H}$. In particular, suppose $\Psi = \lim_{A\to\infty}\Psi_A$ and $\Psi^\prime = \lim_{A\to\infty}\Psi^\prime_A$ for Cauchy sequences\footnote{Strictly speaking, $\Psi$ and $\Psi^\prime$ are equivalence classes of Cauchy sequences.}, $\{\Psi_A\}_{A = 0}^\infty, \{\Psi^\prime_A\}_{A = 0}^\infty \in C_b^\infty$. Then, by the definitions given so far, continuity of inner products and lemma \ref{thm:cbToL2}, 
    \begin{align}
        \langle\Psi, \Psi^\prime\rangle_\mathcal{H}
        &= \lim_{A\to\infty}\lim_{B\to\infty}\langle\Psi_A, \Psi^\prime_B\rangle_{C_b^\infty} = \left\langle\lim_{A\to\infty}\mathfrak{D}(\Psi_A), \lim_{B\to\infty}\mathfrak{D}(\Psi^\prime_B)\right\rangle_{L^2} = \langle \mathfrak{D}(\Psi), \mathfrak{D}(\Psi^\prime)\rangle_{L^2}.
    \end{align}
    An immediate consequence is
    \begin{align}
        ||\mathfrak{D}(\Psi)||_{L^2} = ||\Psi||_\mathcal{H},
        \label{eq:normMatching}
    \end{align}
    which implies that $\mathfrak{D}$ is a continuous/bounded linear operator.
\end{proof}
\begin{theorem}
    $\mathfrak{D}$ is a continuous, linear isomorphism between $\mathcal{H}$ and $L^2$.
    \label{thm:gIsomorphism}
\end{theorem}
Most saliently, the theorem implies $(\gamma^I\nabla_I)^{-1} : L^2 \to \mathcal{H}$ exists.
\begin{proof}
    Linearity is by construction and continuity has already been shown by lemma \ref{thm:GExtension}. Next, suppose $\mathfrak{D}(\Psi) = 0$. Then, by lemma \ref{thm:GExtension}, $0 = ||\mathfrak{D}(\Psi)||_{L^2} = ||\Psi||_{\mathcal{H}} \implies \Psi = 0$ and therefore $\mathfrak{D}$ is injective. It remains to prove surjectivity. Let $\theta$ be an arbitrary element of $L^2$ and define $F_\theta : \mathcal{H} \to \mathbb{C}$ by
    \begin{align}
        F_\theta(\Psi) = \langle\theta, \mathfrak{D}(\Psi)\rangle_{L^2}.
        \label{eq:FTheta}
    \end{align}
    $F_\theta$ is manifestly linear. It is also continuous/bounded because the Cauchy-Schwarz inequality and lemma \ref{thm:GExtension} imply $|F_\theta(\Psi)| = |\langle\theta, \mathfrak{D}(\Psi)\rangle_{L^2}| \leq ||\theta||_{L^2}||\mathfrak{D}(\Psi)||_{L^2} = ||\theta||_{L^2}||\Psi||_{\mathcal{H}}$. Therefore, by the Riesz representation theorem, $\exists\mathcal{Z} \in \mathcal{H}$ such that $F_\theta(\Psi) = \langle\mathcal{Z} , \Psi\rangle_{\mathcal{H}}$. Then, lemma \ref{thm:GExtension} and equation \ref{eq:FTheta} imply
    \begin{align}
       \langle W, \mathfrak{D}(\Psi)\rangle_{L^2} = 0 \,\,\forall \Psi \in \mathcal{H},\,\, \mathrm{where}\,\, W = \theta - \mathfrak{D}(\mathcal{Z}).
       \label{eq:PhiProperty}
    \end{align}
    Using lemma \ref{thm:antisymmetricDerivative}, one can perform a formal integration by parts to get
    \begin{align}
        0 &= \int_\Sigma W^\dagger \mathfrak{D}(\Psi)\mathrm{d}V \\
        &= \int_\Sigma\left(-P_a\overbar{W}\gamma^{ab}D_b\Psi - 3\mathrm{i}kW^\dagger\Psi\right)\mathrm{d}V \\
        &= \int_Sl_an_b\overbar{W}\gamma^{ab}\Psi\mathrm{d}A + \int_\Sigma\left(\gamma^ID_IW + 3\mathrm{i}kW\right)^\dagger\Psi\mathrm{d}V.
    \end{align}
    Let $W = (\phi_\alpha, \bar{\zeta}^{\dot{\alpha}})^T$ and $\Psi = (\psi_\alpha, \bar{\chi}^{\dot{\alpha}})^T$ in terms of two-component spinors. Therefore,
    \begin{align}
        l_an_b\overbar{W}\gamma^{ab}\Psi &= l_an_b\begin{bmatrix}
            -\zeta^\alpha & -\bar{\phi}_{\dot{\alpha}}
        \end{bmatrix}\begin{bmatrix}
            (\sigma^{[a}\tilde{\sigma}^{b]})\indices{_\alpha^\beta} & 0 \\
            0 & (\tilde{\sigma}^{[a}\sigma^{b]})\indices{^{\dot{\alpha}}_{\dot{\beta}}}
        \end{bmatrix}\begin{bmatrix}
            \psi_\beta \\
            \bar{\chi}^{\dot{\beta}}
        \end{bmatrix} \\
        &= \frac{1}{2}\big(n_{\alpha\dot{\alpha}}l^{\beta\dot{\alpha}}\zeta^\alpha\psi_\beta - l_{\alpha\dot{\alpha}}n^{\beta\dot{\alpha}}\zeta^\alpha\psi_\beta + n^{\alpha\dot{\alpha}}l_{\alpha\dot{\beta}}\bar{\phi}_{\dot{\alpha}}\bar{\chi}^{\dot{\beta}} - l^{\alpha\dot{\alpha}}n_{\alpha\dot{\beta}}\bar{\phi}_{\dot{\alpha}}\bar{\chi}^{\dot{\beta}}\big) \\
        &= \sqrt{2}\big(-\zeta_o\psi_\iota - \zeta_\iota\psi_o + \overline{\phi}_o\overline{\chi}_\iota + \overline{\phi}_\iota\overline{\chi}_o\big).
    \end{align}
    Since $\psi_o = \chi_\iota = 0$ on $S \,\,\forall \Psi \in \mathcal{H}$, the formal integration by parts says
    \begin{align}
        0 &= \sqrt{2}\int_S\left(\overline{\phi}_\iota\overline{\chi}_o - \zeta_o\psi_\iota\right)\mathrm{d}A + \int_\Sigma\left(\gamma^ID_IW + 3\mathrm{i}kW\right)^\dagger\Psi\mathrm{d}V.
    \end{align}
    As $\Psi \in \mathcal{H} \supset C_b^\infty$ is arbitrary, it must be that $W$ is a weak solution to $\gamma^ID_IW + 3\mathrm{i}kW = 0$ on $\Sigma$ subject to the boundary conditions, $\phi_\iota = \zeta_o = 0$ on $S$. It is then a technical analytical problem to ascertain whether weak solutions lift to strong solutions in this context. This question was studied in depth by \cite{Bartnik2005, ChruscielBartnik2003}. From the work there, especially theorem 6.4 in \cite{ChruscielBartnik2003}, one can conclude this is indeed the case.

    From here, there is a modified Lichnerowicz identity for $W$ using $\widetilde{\nabla}_aW = D_aW - \mathrm{i}k\gamma_aW$, i.e. $k \mapsto -k$ compared with the original connection, $\nabla$. Then, $\gamma^ID_IW + 3\mathrm{i}kW = 0$ can be re-written as $\gamma^I\widetilde{\nabla}_IW = 0$. The sign of $k$ was never essential in the proof of the Lichnerowicz identity; it merely mattered that $k^2 = -\Lambda/12$.
    Therefore, from the proofs of theorem \ref{thm:lichnerowicz} and lemma \ref{thm:qGHP}, it immediately follows that
    \begin{align}
        0 &= \int_\Sigma(\gamma^I\widetilde{\nabla}_IW)^\dagger\gamma^J\widetilde{\nabla}_J(W)\mathrm{d}V \\
        &= \int_\Sigma\left((\widetilde{\nabla}_IW)^\dagger\widetilde{\nabla}^IW - 4\pi T^{0a}\overbar{W}\gamma_aW\right)\mathrm{d}V - \widetilde{Q}(W),
        \label{eq:nablaTildeLichnerowicz} \\
        \mathrm{where}\,\,\, \widetilde{Q}(W) &= 2\int_S\Big(\phi_\iota\eth\overline{\phi}_o + \overline{\phi}_\iota\bar{\eth}\phi_o - \overline{\zeta}_o\eth\zeta_\iota - \zeta_o\bar{\eth}\overline{\zeta}_\iota + \rho|\phi_o|^2 + \mu|\phi_\iota|^2 + \rho|\zeta_o|^2 + \mu|\zeta_\iota|^2 \nonumber \\
        &\,\,\,\,\,\,\,\,\,\,\,\,\,\,\,\,\,\, - \mathrm{i}k\sqrt{2}\big(\phi_o\zeta_\iota + \phi_\iota\zeta_o - \overline{\phi}_o\overline{\zeta}_\iota - \overline{\phi}_\iota\overline{\zeta}_o\big)\Big)\mathrm{d}A.
    \end{align}
    However, from $\phi_\iota = \zeta_o = 0$ on $S$,
    \begin{align}
        \widetilde{Q}(W) &= 2\int_S\left(\rho|\phi_o|^2 + \mu|\zeta_\iota|^2 - \mathrm{i}k\sqrt{2}\big(\phi_o\zeta_\iota - \overline{\phi}_o\overline{\zeta}_\iota\big)\right)\mathrm{d}A.
    \end{align}
    As in the proof of lemma \ref{def:cbInnerProduct}, let $\mu^\prime = \mu/|z|^2$, $\rho^\prime = |z|^2\rho$, $\phi_o^\prime = \phi_o/z$ and $\zeta^\prime_\iota = z\zeta_\iota$. Again, choose $z = \sqrt[4]{\mu/\rho}$ so that $\mu^\prime = \rho^\prime = -\sqrt{\mu\rho} = -\frac{1}{2}\sqrt{-\theta_l\theta_n} < -k\sqrt{2}$. Therefore,
    \begin{align}
        \widetilde{Q}(W) &= 2\int_S\left(\rho^\prime|\phi^\prime_o|^2 + \mu^\prime|\zeta^\prime_\iota|^2 - \mathrm{i}k\sqrt{2}\big(\phi^\prime_o\zeta^\prime_\iota - \overline{\phi}^\prime_o\overline{\zeta}^\prime_\iota\big)\right)\mathrm{d}A \\
        &= 2\int_S\left((\rho^\prime + k\sqrt{2})|\phi^\prime_o|^2 + (\mu^\prime + k\sqrt{2})|\zeta^\prime_\iota|^2 - k\sqrt{2}|\overline{\phi}^\prime_o + \mathrm{i}\zeta^\prime_\iota|^2\right)\mathrm{d}A 
        \label{eq:qTildeGHP} \\
        &\leq 0.
    \end{align}
    Thus, combined with the dominant energy condition as used in the proof of lemma \ref{def:cbInnerProduct}, every term on the RHS of equation \ref{eq:nablaTildeLichnerowicz} is non-negative.

    Therefore $\widetilde{\nabla}_IW = 0$ and $\widetilde{Q}(W) = 0$. The latter implies $\phi_o = \zeta_\iota = 0$ on $S$ by equation \ref{eq:qTildeGHP} and consequently $W = 0$ on $S$ since $\phi_\iota = \zeta_o = 0$ on $S$ already. In the proof of lemma \ref{def:cbInnerProduct} it was shown $\nabla_I\Psi = 0$ on $\Sigma$ with $\Psi = 0$ on $S$ implies $\Psi = 0$ on $\Sigma$. By the same logic used there, it now follows that $W = 0$ on $\Sigma$. Therefore $\theta = \mathfrak{D}(\mathcal{Z})$ and finally $\mathfrak{D}$ is surjective.
\end{proof}

\section{New quasilocal mass and its positivity}
\label{sec:definition}
Having established all the technical preliminaries, we're now ready to define the new notion of quasilocal mass, the notion of a generic surface and immediately establish that $m(S) \geq 0$ and that every generic surface in AdS has $m(S) = 0$.

A key element of the construction to follow will be Dirac spinors, $\Phi = (\varphi_\alpha,\, \bar{\xi}^{\dot{\alpha}})^T$, satisfying $\overbar{m}^a\nabla_a\Phi = 0$ on $S$. A basis for the solution space will be denoted $\{\Phi^A = (\varphi_\alpha^A,\, \bar{\xi}^{A\dot{\alpha}})^T\}$ and $A, B, ...$ will be indices on this space\footnote{This implicitly assumes the solution space has countable dimension. As will be explained later, this is expected to be the case for generic $S$.}.
\begin{lemma}
    \label{thm:phiGHPEquations}
    Applying the GHP formalism and NP coefficients, $\overbar{m}^a\nabla_a\Phi = 0$ is equivalent to
    \begin{align}
        0 &= \bar{\eth}\varphi_o + \mu\varphi_\iota - \mathrm{i}k\sqrt{2}\,\overline{\xi}_o,
        \label{eq:edthBarAPhi} \\
        0 &= \bar{\eth}\overline{\xi}_\iota - \rho\overline{\xi}_o - \mathrm{i}k\sqrt{2}\,\varphi_\iota, 
        \label{eq:edthBarBBarXi} \\
        0 &= \bar{\eth}\varphi_\iota - \bar{\sigma}\varphi_o \,\,\,\mathrm{and}
        \label{eq:edthBarBPhi}\\
        0 &= \bar{\eth}\overline{\xi}_o + \lambda\overline{\xi}_\iota.
        \label{eq:edthBarABarXi}
    \end{align}
\end{lemma}
\begin{proof}
    In terms of two component spinors, 
    \begin{align}
        \overbar{m}^a\nabla_a\Phi &= \begin{bmatrix}
            \overbar{m}^aD_a\varphi_\alpha \\
            \overbar{m}^aD_a\bar{\xi}^{\dot{\alpha}}
        \end{bmatrix} + \mathrm{i}k\overbar{m}^a\begin{bmatrix}
            0 & (\sigma_a)_{\alpha\dot{\alpha}} \\
            (\tilde{\sigma}_a)^{\dot{\alpha}\alpha} & 0
        \end{bmatrix}\begin{bmatrix}
            \varphi_\alpha \\
            \bar{\xi}^{\dot{\alpha}}
        \end{bmatrix} = \begin{bmatrix}
            \bar{\delta}\varphi_\alpha - \mathrm{i}k\sqrt{2}\,\overline{\xi}_oo_\alpha \\
            \bar{\delta}\bar{\xi}^{\dot{\alpha}} - \mathrm{i}k\sqrt{2}\varphi_\iota\bar{\iota}^{\dot{\alpha}}
        \end{bmatrix}.
        \label{eq:deltaBarPhiDiracNotation}
    \end{align}
    Contracting the top half with $o^\alpha$, applying the NP coefficients from appendix \ref{sec:identities} and equations \ref{eq:edth} \& \ref{eq:edthBar} yields
    \begin{align}
        0 &= o^\alpha\left(\bar{\delta}\varphi_\alpha - \mathrm{i}k\sqrt{2}\,\overline{\xi}_oo_\alpha\right) \\
        &= o^\alpha\left(\bar{\delta}\left(\varphi_oo_\alpha + \varphi_\iota\iota_\alpha\right) - \mathrm{i}k\sqrt{2}\,\overline{\xi}_oo_\alpha\right) \\
        &= \sqrt{2}\left(\bar{\sigma}\varphi_o - \bar{\eth}\varphi_\iota\right),
    \end{align}
    which proves equation \ref{eq:edthBarBPhi}. Similarly,
    \begin{align}
        0 &= \iota^\alpha\left(\bar{\delta}\varphi_\alpha - \mathrm{i}k\sqrt{2}\,\overline{\xi}_oo_\alpha\right) 
        = \sqrt{2}\left(\bar{\eth}\varphi_o + \mu\varphi_\iota - \mathrm{i}k\sqrt{2}\,\overline{\xi}_o\right), \\
        0 &= \bar{o}_{\dot{\alpha}}\left(\bar{\delta}\bar{\xi}^{\dot{\alpha}} - \mathrm{i}k\sqrt{2}\varphi_\iota\bar{\iota}^{\dot{\alpha}}\right) = \sqrt{2}\left(\bar{\eth}\overline{\xi}_\iota - \rho\overline{\xi}_o - \mathrm{i}k\sqrt{2}\varphi_\iota\right)\,\,\,\mathrm{and} \\
        0 &= \bar{\iota}_{\dot{\alpha}}\left(\bar{\delta}\bar{\xi}^{\dot{\alpha}} - \mathrm{i}k\sqrt{2}\varphi_\iota\bar{\iota}^{\dot{\alpha}}\right) = -\sqrt{2}\left(\bar{\eth}\overline{\xi}_o + \lambda\overline{\xi}_\iota\right)
    \end{align}
    prove the remaining three equations.
\end{proof}
\begin{definition}[$Q^{AB}$]
    \label{def:qAB}
    Define the hermitian matrix, $Q^{AB}$, by
    \begin{align}
        Q^{AB} &= 4\int_S\Big(\rho\overline{\varphi}_o^A\varphi_o^B + \mu \xi_\iota^A\overline{\xi}_\iota^B - \rho \xi_o^A\overline{\xi}_o^B - \mu\overline{\varphi}_\iota^A\varphi_\iota^B \nonumber \\
        &\,\,\,\,\,\,\,\,\,\,\,\,\,\,\,\,\,\, + \mathrm{i}k\sqrt{2}\big(\xi_\iota^A\varphi_o^B - \overline{\varphi}_o^A\overline{\xi}_\iota^B - \xi_o^A\varphi_\iota^B + \overline{\varphi}_\iota^A\overline{\xi}_o^B\big)\Big)\mathrm{d}A.
    \end{align}
\end{definition}
\begin{theorem}
    If the dominant energy condition holds on $\Sigma$ and the null expansions on $S$ satisfy $\theta_l > 0$, $\theta_n < 0$ \& $\theta_l\theta_n < -8k^2$, then $Q^{AB}$ is a non-negative definite matrix.
    \label{thm:qNonNegative}
\end{theorem}
\begin{proof}
    From lemma \ref{thm:qGHP},
    \begin{align}
        Q(\Phi) &= 4\int_S\Big(\varphi_\iota\eth\overline{\varphi}_o + \overline{\varphi}_\iota\bar{\eth}\varphi_o - \overline{\xi}_o\eth\xi_\iota - \xi_o\bar{\eth}\overline{\xi}_\iota + \rho|\varphi_o|^2 + \mu|\varphi_\iota|^2 + \rho|\xi_o|^2 + \mu|\xi_\iota|^2 \nonumber \\
        &\,\,\,\,\,\,\,\,\,\,\,\,\,\,\,\,\,\, + \mathrm{i}k\sqrt{2}\big(\varphi_o\xi_\iota + \varphi_\iota\xi_o - \overline{\varphi}_o\overline{\xi}_\iota - \overline{\varphi}_\iota\overline{\xi}_o\big)\Big)\mathrm{d}A.
    \end{align}
    From equations \ref{eq:edthBarAPhi} and \ref{eq:edthBarBBarXi}, this reduces to
    \begin{align}
        Q(\Phi) &= 4\int_S\Big(\rho|\varphi_o|^2 - \mu|\varphi_\iota|^2 - \rho|\xi_o|^2 + \mu|\xi_\iota|^2  + \mathrm{i}k\sqrt{2}\big(\varphi_o\xi_\iota - \xi_o\varphi_\iota - \overline{\varphi}_o\overline{\xi}_\iota + \overline{\xi}_o\overline{\varphi}_\iota\big)\Big)\mathrm{d}A.
        \label{eq:qPhiGHP}
    \end{align}
    Let $\mathcal{Z} = (\phi_\alpha, \bar{\zeta}^{\dot{\alpha}})^T$ be any Dirac spinor on $\Sigma$ with sufficient regularity so that $\gamma^I\nabla_I\mathcal{Z} \in L^2$. Furthermore, choose $\mathcal{Z}$ to have $\phi_o = \varphi_o$ and $\zeta_\iota = \xi_\iota$ on $S$. Therefore, by theorem \ref{thm:gIsomorphism}, $\exists \Psi^\prime \in \mathcal{H}$ such that $\mathfrak{D}(\Psi^\prime) = -\gamma^I\nabla_I\mathcal{Z}$. Thus $\Psi = \Psi^\prime + \mathcal{Z}$ satisfies $\gamma^I\nabla_I\Psi = 0$ and then by definition \ref{def:q},
    \begin{align}
        Q(\Psi) &= \int_{\Sigma}\left(\nabla_I(\Psi)^\dagger\nabla^I\Psi - 4\pi T^{0a}\overline{\Psi}\gamma^a\Psi\right)\mathrm{d}V \geq 0,
        \label{eq:qPsiSolvesDirac}
    \end{align}
    where the first term is manifestly non-negative and the second term is non-negative by the dominant energy condition. Furthermore, since every element, $\Psi^\prime \in \mathcal{H}$, has $\psi^\prime_o = \chi^\prime_\iota = 0$ on $S$ by construction, it follows that $\Psi$ has $\psi_o = \varphi_o$ and $\chi_\iota = \xi_\iota$ on $S$.

    Therefore, by definition \ref{def:q}, lemma \ref{thm:qGHP} and the fact all the derivatives in lemma \ref{thm:qGHP} are tangent to $S$, $Q(\Psi)$ can also be written as
    \begin{align}
        Q(\Psi) &= 4\int_S\Big(\psi_\iota\eth\overline{\varphi}_o + \overline{\psi}_\iota\bar{\eth}\varphi_o - \overline{\chi}_o\eth\xi_\iota - \chi_o\bar{\eth}\overline{\xi}_\iota + \rho|\varphi_o|^2 + \mu|\psi_\iota|^2 + \rho|\chi_o|^2 + \mu|\xi_\iota|^2 \nonumber \\
        &\,\,\,\,\,\,\,\,\,\,\,\,\,\,\,\,\,\, + \mathrm{i}k\sqrt{2}\big(\varphi_o\xi_\iota + \psi_\iota\chi_o - \overline{\varphi}_o\overline{\xi}_\iota - \overline{\psi}_\iota\overline{\chi}_o\big)\Big)\mathrm{d}A.
    \end{align}
    Then, from equations \ref{eq:edthBarAPhi} and \ref{eq:edthBarBBarXi},
    \begin{align}
        Q(\Psi) 
        &= 4\int_S\Big(\mu\big(-\psi_\iota\overline{\varphi}_\iota - \overline{\psi}_\iota\varphi_\iota + |\psi_\iota|^2 + |\xi_\iota|^2\big) + \rho\big(-\overline{\chi}_o\xi_o - \chi_o\overline{\xi}_o + |\varphi_o|^2 + |\chi_o|^2\big) \nonumber \\
        &\,\,\,\,\,\,\,\,\,\,\,\,\,\,\,\,\,\, + \mathrm{i}k\sqrt{2}\big(-\psi_\iota\xi_o + \overline{\psi}_\iota\overline{\xi}_o + \overline{\chi}_o\overline{\varphi}_\iota - \chi_o\varphi_\iota + \varphi_o\xi_\iota + \psi_\iota\chi_o - \overline{\varphi}_o\overline{\xi}_\iota - \overline{\psi}_\iota\overline{\chi}_o\big)\Big)\mathrm{d}A.
    \end{align}
    Therefore re-writing equation \ref{eq:qPhiGHP} in terms of $Q(\Psi)$ yields
    \begin{align}
        Q(\Phi) &= 4\int_S\Big(-\mu\big(|\varphi_\iota|^2 - \psi_\iota\overline{\varphi}_\iota - \overline{\psi}_\iota\varphi_\iota + |\psi_\iota|^2\big) - \rho\big(|\xi_o|^2 - \overline{\chi}_o\xi_o - \chi_o\overline{\xi}_o + |\chi_o|^2\big) \nonumber \\
        &\,\,\,\,\,\,\,\,\,\,\,\,\,\,\,\,\,\, - \mathrm{i}k\sqrt{2}\big(-\psi_\iota\xi_o + \overline{\psi}_\iota\overline{\xi}_o + \overline{\chi}_o\overline{\varphi}_\iota - \chi_o\varphi_\iota \nonumber \\
        &\,\,\,\,\,\,\,\,\,\,\,\,\,\,\,\,\,\,\,\,\,\,\,\,\, + \psi_\iota\chi_o - \overline{\psi}_\iota\overline{\chi}_o + \xi_o\varphi_\iota - \overline{\xi}_o\overline{\varphi}_\iota\big)\Big)\mathrm{d}A + Q(\Psi) \\
        &= 4\int_S\Big(- \mathrm{i}k\sqrt{2}\big((\xi_o - \chi_o)(\varphi_\iota - \psi_\iota) - (\overline{\xi}_o - \overline{\chi}_o)(\overline{\varphi}_\iota - \overline{\psi}_\iota)\big) \nonumber \\
        &\,\,\,\,\,\,\,\,\,\,\,\,\,\,\,\,\,\,\,\,\, -\mu|\varphi_\iota - \psi_\iota|^2 - \rho|\xi_o - \chi_o|^2 \Big)\mathrm{d}A + Q(\Psi).
    \end{align}
    As done previously in section \ref{sec:analysis}, let $\mu^\prime = \mu/|z|^2$, $\rho^\prime = |z|^2\rho$, $\xi_o^\prime = \xi_o/z$, $\chi^\prime_o = \chi_o/z$, $\varphi^\prime_\iota = z\varphi_\iota$ and $\psi^\prime_\iota = z\psi_\iota$.  Again, choose $z = \sqrt[4]{\mu/\rho}$ so that $\mu^\prime = \rho^\prime = -\sqrt{\mu\rho} = -\frac{1}{2}\sqrt{-\theta_l\theta_n} < -k\sqrt{2}$. Hence,
    \begin{align}
        Q(\Phi) &= 4\int_S\Big(- \mathrm{i}k\sqrt{2}\big((\xi^\prime_o - \chi^\prime_o)(\varphi^\prime_\iota - \psi^\prime_\iota) - (\overline{\xi}^\prime_o - \overline{\chi}^\prime_o)(\overline{\varphi}^\prime_\iota - \overline{\psi}^\prime_\iota)\big) \nonumber \\
        &\,\,\,\,\,\,\,\,\,\,\,\,\,\,\,\,\,\,\,\,\, -\mu^\prime|\varphi^\prime_\iota - \psi^\prime_\iota|^2 - \rho^\prime|\xi^\prime_o - \chi^\prime_o|^2 \Big)\mathrm{d}A + Q(\Psi) \\
        &= 4\int_S\Big(\sqrt{2}k|\xi^\prime_o - \chi^\prime_o + \mathrm{i}\overline{\varphi}^\prime_\iota - \mathrm{i}\overline{\psi}^\prime_\iota|^2 \nonumber \\
        &\,\,\,\,\,\,\,\,\,\,\,\,\,\,\,\,\,\,\,\,\, -(\mu^\prime + \sqrt{2}k)|\varphi^\prime_\iota - \psi^\prime_\iota|^2 - (\rho^\prime +\sqrt{2}k)|\xi^\prime_o - \chi^\prime_o|^2 \Big)\mathrm{d}A + Q(\Psi) \\
        &\geq 0.
        \label{eq:qPhiNonNegative}
    \end{align}
    Since $\{\Phi^A\}$ is a basis for the solution space to $\overbar{m}^a\nabla_a\Phi = 0$, one can let $\Phi = c_A\Phi^A$ for any constants, $c_A$. Hence, $\varphi
    _o = c_A\varphi_o^A$, $\varphi_\iota = c_A\varphi_\iota^A$, $\overline{\xi}_o = c_A\overline{\xi}_o^A$ and $\overline{\xi}_\iota = c_A\overline{\xi}_\iota^A$.
    Finally, definition \ref{def:qAB}, equation \ref{eq:qPhiGHP} and equation \ref{eq:qPhiNonNegative} imply
    \begin{align}
        0 \leq Q(\Phi) = \bar{c}_AQ^{AB}c_B.
    \end{align}
    Since $c_A$ are arbitrary, it must be that $Q^{AB}$ is non-negative definite.
\end{proof}
While this theorem achieves a manifestly non-negative object, some auxiliary constructions are still required to extract a mass from $Q^{AB}$.
\begin{definition}[$T^{AB}$]
    \label{def:tAB}
    Define the matrix, $T^{AB}$, by 
    \begin{align}
        T^{AB} &= \varepsilon^{\alpha\beta}\varphi_\alpha^A\varphi_\beta^B - \varepsilon^{\dot{\alpha}\dot{\beta}}\bar{\xi}_{\dot{\alpha}}^A\bar{\xi}_{\dot{\beta}}^B = \sqrt{2}\left(\varphi_o^A\varphi_\iota^B - \varphi_o^B\varphi_\iota^A + \overline{\xi}_o^B\overline{\xi}_\iota^A - \overline{\xi}_o^A\overline{\xi}_\iota^B\right).
    \end{align}
\end{definition}
The notion of a surface, $S$, being ``generic" can now finally be stated precisely. 
\begin{definition}[Generic - $T^{AB}$ form]
    The surface, $S$, is called generic if and only if $T^{AB}$ is invertible.
    \label{def:genericT}
\end{definition}
\begin{definition}[Generic - $\Phi^A$ form]
    The surface, $S$, is said to be generic if and only if the solution space to $\overbar{m}^a\nabla_a\Phi = 0$ on $S$ is four (complex) dimensional and the basis, $\{\Phi^A\}_{A = 1}^4$, is pointwise linearly independent at least at one point of $S$.
    \label{def:genericPhi}
\end{definition}
It will be shown later that the $\Phi^A$ version of generic implies the $T^{AB}$ version, although only the $T^{AB}$ version will be needed for defining $m(S)$. More importantly though, surfaces generic in name should be generic in practice too. For the $T^{AB}$ form, it could be argued that since the set of singular $n\times n$ matrices are measure zero in the set of all $n\times n$ matrices, this is indeed a valid notion of generic. However, it's not obvious the solution space is finite dimensional and this argument doesn't consider the possibility there is something specific to this situation precluding $T^{AB}$'s invertibility. Furthermore, the examples considered in sections \ref{sec:schwarzschild} and \ref{sec:asymptotics} either satisfy both notions of  generic or neither notion of generic. Hence, it's unclear whether $m(S)$ constructed on a surface satisfying the $T^{AB}$ form, but not the $\Phi^A$ form, of generic has physical meaning beyond simple mathematical validity. Finally, from a practical point of view, one would like to know what size of matrix to expect for $T^{AB}$ - and for that matter, $Q^{AB}$. As defined so far, they could be of arbitrarily large size, maybe even infinitely large. Fortunately, at least for topologically spherical $S$, there are reasons to believe the $\Phi^A$ form is also a valid notion of generic, implying $T^{AB}$ is only a $4\times 4$ matrix.

It is known - e.g. from section 8.2.2 of \cite{Szabados2009} - that $\bar{\delta}$ is an elliptic operator and the compactness of $S$ then guarantees $\bar{\delta}$ has finite dimensional kernel. Then, it is also known \cite{Szabados2009} that $\bar{\delta}$'s index (dimension of kernel minus dimension of cokernel) is $4(1 - g)$ when $S$ has genus, $g$. The difference between $\overbar{m}^a\nabla_a$ - the actual operator of interest - and $\bar{\delta}$ is $\mathrm{i}k\overbar{m}^a\gamma_a$, which is a compact operator since $S$ is compact and $\mathrm{i}k\overbar{m}^a\gamma_a$ is just a $4\times 4$ matrix. Therefore by Fredholm theory, $\mathrm{index}(\overbar{m}^a\nabla_a) = \mathrm{index}(\bar{\delta}) = 4(1 - g)$. Thus, if $S$ is diffeomorphic to a sphere, then $\overbar{m}^a\nabla_a\Phi = 0$ must have at least four linearly independent solutions.

In the spherical examples of sections \ref{sec:schwarzschild} and \ref{sec:asymptotics}, there happen to be precisely four linearly independent solutions. From a similar situation, Penrose then argues \cite{Penrose1982} as long as $S$ is not too far from ``canonical" situations - such as the examples to be considered - there would still remain precisely four linearly independent solutions. At least for spherical $S$, this justifies the first half of the generic definition in $\Phi^A$ form. For non-spherical $S$, the situation is far less constrained and it cannot immediately be said whether either definition of generic is actually realistic. Section \ref{sec:torus} is devoted to studying two examples with toroidal $S$. In the first, it will be shown $\overbar{m}^a\nabla_a\Phi = 0$ has two linearly independent solutions and the corresponding $T^{AB}$ is just zero, while in the second, $\overbar{m}^a\nabla_a\Phi = 0$ won't even have a single non-zero solution. Hence both definitions of generic fail. However, the wider implications of those examples are unclear.

The second half of definition \ref{def:genericPhi} is motivated by a possibility that occurs in the Dougan-Mason definition, where one needs to solve the analogous equation, $\bar{\delta}\varphi_\alpha = 0$. It turns out there exist ``exceptional" surfaces - bifurcate Killing surfaces are one example - where there are two solutions to $\bar{\delta}\varphi_\alpha = 0$ (the maximum expected or desired in that context) which are linearly independent as functions despite being pointwise linearly dependent at every point of $S$. The Dougan-Mason mass cannot be defined on such surfaces because the analogue of $T^{AB}$ just becomes zero. However, based on considerations of holomorphic spin bundles, Dougan and Mason argue such surfaces really are exceptional and not generic. Similarly, definition \ref{def:genericPhi} insists $\{\Phi^A\}_{A = 1}^4$ are pointwise linearly independent at least at one point of $S$ for $S$ to be called generic in the $\Phi^A$ sense.
\begin{lemma}
    $T^{AB}$ is antisymmetric and constant on $S$. Furthermore, the notion of generic in definition \ref{def:genericPhi} implies the notion of generic in definition \ref{def:genericT}.
    \label{thm:tProperties}
\end{lemma}
\begin{proof}
    Antisymmetry follows directly from the definition. Next, observe that by equation \ref{eq:deltaBarPhiDiracNotation},
    \begin{align}
        \bar{\delta}T^{AB}
        &= \varepsilon^{\alpha\beta}\bar{\delta}(\varphi_\alpha^A)\varphi_\beta^B + \varepsilon^{\alpha\beta}\varphi_\alpha^A\bar{\delta}\varphi_\beta^B - \varepsilon^{\dot{\alpha}\dot{\beta}}\bar{\delta}(\bar{\xi}_{\dot{\alpha}}^A)\bar{\xi}_{\dot{\beta}}^B - \varepsilon^{\dot{\alpha}\dot{\beta}}\bar{\xi}_{\dot{\alpha}}^A\bar{\delta}\bar{\xi}_{\dot{\beta}}^B \\
        &= 2\mathrm{i}k\left(\overline{\xi}_o^A\varphi_\iota^B - \overline{\xi}_o^B\varphi_\iota^A + \varphi_\iota^A\overline{\xi}_o^B - \varphi_\iota^B\overline{\xi}_o^A\right) \\
        &= 0.
    \end{align}
    Therefore for each $A$ and $B$, $T^{AB}$ is a holomorphic function on $S$. Then, since $S$ is compact, Liouville's theorem implies $T^{AB}$ is constant on $S$.

    To prove invertibility, it's easier to work in Dirac spinor notation. With the charge conjugation matrix given in appendix \ref{sec:conventions}, observe that
    \begin{align}
        (\Phi^A)^TC^{-1}\Phi^B &= \begin{bmatrix}
            \varphi^A_\alpha & \bar{\xi}^{A\dot{\alpha}}
        \end{bmatrix}\begin{bmatrix}
            \varepsilon^{\alpha\beta} & 0 \\
            0 & \varepsilon_{\dot{\alpha}\dot{\beta}}
        \end{bmatrix}\begin{bmatrix}
            \varphi_\beta^B \\
            \bar{\xi}^{B\dot{\beta}}
        \end{bmatrix} = \varphi^A_\alpha\varepsilon^{\alpha\beta}\varphi_\beta^B + \bar{\xi}^{A\dot{\alpha}}\varepsilon_{\dot{\alpha}\dot{\beta}}\bar{\xi}^{B\dot{\beta}} = T^{AB}.
        \label{eq:tDirac}
    \end{align}
    Let $v_A$ be a vector in the nullspace of $T^{AB}$, i.e. $T^{AB}v_B = 0$. Let $\mathcal{Z} = v_A\Phi^A$. Then, 
    \begin{align}
        T^{AB}v_B = 0 \iff (\Phi^A)^TC^{-1}\mathcal{Z} = 0 \iff w_A(\Phi^A)^TC^{-1}\mathcal{Z} = 0
    \end{align}
    for any vector, $w_A$. Definition \ref{def:genericPhi} says there are four different $\Phi^A$ and they are pointwise linearly independent at least at one point, say $p$, on $S$. Since Dirac spinors also have four components, $\{\Phi^A\}_{A = 1}^4$ must form a pointwise basis at $p$. Hence, $w_A\Phi^A$ can be any Dirac spinor at $p$, which then implies $C^{-1}\mathcal{Z}|_p = 0$. But $C^{-1}$ is invertible, so it must be that $\mathcal{Z}|_p = 0$. However, then $v_A = 0$ by the linear independence of $\{\Phi^A\}_{A = 1}^4$ at $p$, which is equivalent to $T^{AB}$ having trivial nullspace.
\end{proof}
Before $T^{AB}$ can be put to use in extracting information from $Q^{AB}$, one more auxiliary result - which happens to just be a basic fact in linear algebra - is required.
\begin{lemma}
    For any non-negative definite, hermitian matrix, $H$, and antisymmetric matrix, $A$, $\tr(HA\overbar{H}\bar{A})$ is real and $\tr(HA\overbar{H}\bar{A}) \leq 0$.
    \label{thm:funny}
\end{lemma}
\begin{proof}
    In this proof all indices will be downstairs and all summations will be written explicitly. Every hermitian matrix is orthogonally diagonalisable and has real eigenvalues. Therefore, 
    $\exists$ vectors, $\{v_{(A)}\}$, such that $v_{(A)}^\dagger v_{(B)} = \delta_{AB}$ and $Hv_{(A)} = \lambda_Av_{(A)}$ for some $\lambda_A \in \mathbb{R}$. $H$ being non-negative definite implies $\lambda_A \geq 0\,\,\forall A$. Let $U_{AB} = v_{(B)A}$ so orthogonal diagonalisation says 
    \begin{align}
        (U^\dagger HU)_{AB} &= \sum_{C, D}(U^\dagger)_{AC}H_{CD}U_{DB} \sum_{C, D}\bar{v}_{(A)C}H_{CD}v_{(B)D} = \delta_{AB}\lambda_B = D_{AB} \\
        \iff H_{AB} &= \sum_{C, D}U_{AC}D_{CD}(U^\dagger)_{DB} = \sum_{C, D}v_{(C)A}\delta_{CD}\lambda_D\bar{v}_{(D)B} = \sum_C\lambda_Cv_{(C)A}\bar{v}_{(C)B}.
    \end{align}
    Then, the quantity of interest is
    \begin{align}
        \tr(HA\overbar{H}\bar{A}) &= -\sum_{A, B, C, D}H_{AB}A_{BC}\overbar{H}_{CD}\bar{A}_{AD} \\
        &= -\sum_{A, B, C, D, E, F} \lambda_Ev_{(E)A}\bar{v}_{(E)B}A_{BC}\lambda_F\bar{v}_{(F)C}v_{(F)D}\bar{A}_{AD} \\
        &= -\sum_{E, F}\lambda_E\lambda_F|v_{(E)}^\dagger A\bar{v}_{(F)}|^2,
    \end{align}
    which is manifestly real and non-positive.
\end{proof}
\begin{definition}[Quasilocal mass]
    \label{def:quasilocalMass}
    Suppose the dominant energy condition holds on $\Sigma$, the null expansions on $S$ satisfy $\theta_l > 0$, $\theta_n < 0$ \& $\theta_l\theta_n < -8k^2$ and $S$ is generic (either definition). Then, construct $Q^{AB}$ and $T^{AB}$ by definitions \ref{def:qAB} \& \ref{def:tAB} and define the quasilocal mass, $m(S)$, to be
    \begin{align}
        m(S) &= \frac{1}{16\pi}\sqrt{-\tr(QT^{-1}\overbar{Q}\overbar{T}^{-1})}.
    \end{align}
\end{definition}
Theorem \ref{thm:qNonNegative}, lemma \ref{thm:tProperties} and lemma \ref{thm:funny} ensure $m(S)$ is well-defined and manifestly non-negative. Furthermore, $m(S)$ is independent of the choice of basis, $\{\Phi^A\}$, as follows. Define $\Phi^{\prime A} = B\indices{^A_B}\Phi^B$ for a constant, invertible matrix, $B$. Then, by definitions \ref{def:qAB} and \ref{def:tAB},
\begin{align}
    Q^{\prime AB} &= \bar{B}\indices{^A_C}Q^{CD}B\indices{^B_D} \iff Q^\prime = \bar{B}QB^T \,\,\,\mathrm{and} \\
    T^{\prime AB} &= B\indices{^A_C}T^{CD}B\indices{^B_D} \iff T^\prime = BTB^T.
\end{align}
Thus, the object in $m(S)$ transforms as
\begin{align}
    \tr(Q^\prime(T^\prime)^{-1}\overbar{Q}^\prime(\overbar{T}^\prime)^{-1}) &= \tr(\bar{B}QB^TB^{-T}T^{-1}B^{-1}B\overbar{Q}\bar{B}^T\bar{B}^{-T}\overbar{T}^{-1}\bar{B}^{-1}) = \tr(QT^{-1}\overbar{Q}\overbar{T}^{-1}).
\end{align}
\begin{lemma}
    $m(S) = 0$ for every surface, $S$, in AdS that is generic in the $\Phi^A$ sense.
\end{lemma}
\begin{proof}
    With this notion of generic $\exists$ exactly four linearly independent solutions to $\overbar{m}^a\nabla_a\Phi = 0$. However, AdS already has a four dimensional space of Killing spinors, i.e. solutions to $\nabla_a\varepsilon_k = 0$. Therefore, since $\nabla_a\varepsilon_k = 0$ is a stronger condition, one can use the Killing spinors of AdS as $\{\Phi^A\}_{A = 1}^4$. Then, $\Phi = \varepsilon_k$ and $\nabla_a\varepsilon_k = 0 \implies E^{ab}(\Phi) = 0 \implies Q(\Phi) = 0 \implies m(S) = 0$.
\end{proof}
It's natural to consider the converse, i.e. study the implications of $m(S) = 0$. This problem is considerably harder even in the asymptotically flat or asymptotically hyperbolic context - see \cite{Hirsch2024a, Hirsch2024b} for recent progress in those cases - and will not be considered in this work.

The new definition of quasilocal mass is closest in spirit to Penrose's definition, albeit there is no need for twistors. In particular, $Q^{AB}$ is analogous to Penrose's ``kinematical twistor" - see the material around equation 23 in \cite{Penrose1982} - while $T^{AB}$ is analogous to his surface ``infinity twistor" - see the discussion between equations 25 and 26 in \cite{Penrose1982}. Meanwhile, the present definition is also closely related to the Dougan-Mason mass. When $\Lambda = 0$, the left-handed and right-handed sectors of all the equations decouple, meaning it suffices to simply set the right-handed sector to zero. Then, $A, B, \cdots$ only run $1, 2$. Thus, $T^{AB}$ can be normalised to $\varepsilon^{AB}$ and one can use it to manipulate two-component spinors with $Q^{AB}$ now viewed as $P^{\dot{A}A}$, a 4-momentum converted to two-component spinors. Then, the present definition says
\begin{align}
    -256\pi^2m(S)^2 &= \tr(QT^{-1}\overbar{Q}\overbar{T}^{-1}) = Q^{AB}T_{BC}\overbar{Q}^{CD}\overbar{T}_{DA} \equiv P^{\dot{A}A}\varepsilon_{AB}P^{B\dot{B}}\varepsilon_{\dot{B}\dot{A}},
\end{align} 
which is the Dougan-Mason mass (up to normalisation). However, since Dougan and Mason have a full energy-momentum vector, $P^{\dot{A}A}$, they are able to further decompose $m(S)$ into a quasilocal energy and quasilocal linear momentum. This decomposition is lost in the present definition - as it is in Penrose's definition when $S$ is away from $\mathcal{I}$. While the technical reason is simply that $A, B, \cdots$ run over four indices, instead of two, a more physical reason could be the difference between the Casimir operators of $\mathfrak{o(3, 2)}$ and $\mathfrak{iso(3, 1)}$, as will be discussed further at the end of section \ref{sec:asymptotics}.

\section{Highly symmetric examples}
\label{sec:schwarzschild}
For an arbitrary surface, $S$, the quasilocal mass of definition \ref{def:quasilocalMass} will likely be very difficult, if not impossible, to calculate analytically. However, if the surface has a high degree of symmetry, then more progress can be made. In section \ref{sec:spherical} the focus will be spherically symmetric spacetimes, where it will be shown definition \ref{def:quasilocalMass} reduces to the Misner-Sharp mass\footnote{The Misner-Sharp mass is usually taken as the standard mass for spherically symmetric spacetimes \cite{Szabados2009}.} \cite{Misner1964} of such spacetimes\footnote{While this could appear to be merely a sanity check, in fact it is non-trivial. For example, the Brown-York mass \cite{Brown1993} does not agree with the Misner-Sharp mass and in fact produces $m(S_r^2) = r(1 - \sqrt{1 - 2M/r})$ in the Schwarzschild spacetime (with $\Lambda = 0$) despite being physically very well-motivated.}. Likewise, section \ref{sec:torus} explores two examples with toroidal symmetry, where it will turn out that a number of assumptions required for definition \ref{def:quasilocalMass} don't hold. The canonical examples of spacetimes with such high symmetry are the Schwarzschild spacetime and its variations, described by the metric,
\begin{align}
    g &= -\left(c - \frac{2M}{r} + 4k^2\right)\mathrm{d}t\otimes\mathrm{d}t + \frac{\mathrm{d}r\otimes\mathrm{d}r}{c - 2M/r + 4k^2} + r^2g^{(c)},
    \label{eq:schwarzschildGeneral}
\end{align}
where $c = 1$, $0$ or $-1$ and $g^{(c)}$ is the standard metric on the round 2-sphere, the 2-torus or a compactified 2D hyperbolic space respectively.

\subsection{Spherical symmetry}
\label{sec:spherical}
Given there is a heavy reliance on null normals in the NP and GHP formalisms, it will be be easiest to study a general, spherically symmetric spacetime by deploying double null coordinates. In particular, for any spherically symmetric spacetime, let $r$ be the area-radius function and let $u$ \& $v$ be null coordinates normal to the symmetry spheres, $S_r^2$. Then, in such ``double null" coordinates, spherical symmetry dictates the metric is
    \begin{align}
        g &= -\Omega(u, v)^2(\mathrm{d}u\otimes\mathrm{d}v + \mathrm{d}v\otimes\mathrm{d}u) + r(u, v)^2g_{S^2},
    \end{align}
for some function, $\Omega(u, v)$. Without loss of generality assume $u$ is outgoing and $v$ is ingoing, i.e. $\partial_ur > 0$ and $\partial_vr< 0$.

For any $S_r^2$ in this spacetime, a natural NP tetrad is
\begin{align}
    l &= \frac{1}{\Omega}\frac{\partial}{\partial u},\,\, n = \frac{1}{\Omega}\frac{\partial}{\partial v} \,\,\mathrm{and}\,\, m = \frac{1}{r\sqrt{2}}\left(\frac{\partial}{\partial\theta} + \frac{\mathrm{i}}{\sin(\theta)}\frac{\partial}{\partial\phi}\right).
\end{align}
For the chosen tetrad, one finds 
\begin{align}
    \sigma = \lambda = 0, \,\,\, \rho = -\frac{\partial_ur}{\Omega r},\,\,\, \mu = \frac{\partial_vr}{\Omega r} \,\,\,\mathrm{and}\,\,\, \beta = -\alpha = \frac{1}{2\sqrt{2}r}\cot(\theta)
    \label{eq:schwarzschildNP}
\end{align}
by direct calculation.
\begin{theorem}
    The general solution to $\overbar{m}^a\nabla_a\Phi$ on $S_r^2$ is 
    \begin{align}
        \xi_o &= \bar{c}_1\left({}_{1/2}Y_{1/2,-1/2}\right) + \bar{c}_2\left({}_{1/2}Y_{1/2,1/2}\right), \\
        \varphi_\iota &= c_3\left({}_{-1/2}Y_{1/2,-1/2}\right) + c_4\left({}_{-1/2}Y_{1/2,1/2}\right), \\
        \varphi_o &= -\bigg(\frac{\sqrt{2}}{\Omega}\partial_v(r)c_3 + 2\mathrm{i}krc_2\bigg)\left({}_{1/2}Y_{1/2,-1/2}\right) - \bigg(\frac{\sqrt{2}}{\Omega}\partial_v(r)c_4 - 2\mathrm{i}krc_1\bigg)\left({}_{1/2}Y_{1/2,1/2}\right) \,\,\,\mathrm{and} \\
        \xi_\iota &= \bigg(\frac{\sqrt{2}}{\Omega}\partial_u(r)\bar{c}_1 + 2\mathrm{i}kr\bar{c}_4\bigg)\left({}_{-1/2}Y_{1/2,-1/2}\right) + \bigg(\frac{\sqrt{2}}{\Omega}\partial_u(r)\bar{c}_2 - 2\mathrm{i}kr\bar{c}_3\bigg)\left({}_{-1/2}Y_{1/2,1/2}\right),
    \end{align}
    where $c_A$ are arbitrary constants and $\left({}_{s}Y_{jm}\right)$ are spin-weighted spherical harmonics\footnote{The exact expressions for the four spin-weighted spherical harmonics used are listed in appendix \ref{sec:conventions}.}.
    \label{thm:phiSphere}
\end{theorem}
\begin{proof}
    Let $\eth_s$ and $\bar{\eth}_s$ be differential operators that act on functions on the sphere, $F$, by
    \begin{align}
        \eth_sF &= s\cot(\theta)F - \left(\frac{\partial}{\partial\theta} + \frac{\mathrm{i}}{\sin(\theta)}\frac{\partial}{\partial\phi}\right)F \,\,\,\mathrm{and} \\
        \bar{\eth}_sF &= -s\cot(\theta)F - \left(\frac{\partial}{\partial\theta} - \frac{\mathrm{i}}{\sin(\theta)}\frac{\partial}{\partial\phi}\right)F. 
    \end{align}
    Then, since $\varphi_o$ \& $\xi_o$ are type-$(0, -1)$ and $\varphi_\iota$ \& $\xi_\iota$ are type-$(0, 1)$ in the GHP formalism, the chosen tetrad and the NP coefficients in equation \ref{eq:schwarzschildNP} imply the equations of lemma \ref{thm:phiGHPEquations} can be written as
    \begin{align}
        0 &= \bar{\eth}_{1/2}\varphi_o  - \frac{\sqrt{2}}{\Omega}\partial_v(r)\varphi_\iota + 2\mathrm{i}rk\overline{\xi}_o,
        \label{eq:edthHalfBarAPhi} \\
        0 &= \eth_{-1/2}\xi_\iota - \frac{\sqrt{2}}{\Omega}\partial_u(r)\xi_o - 2\mathrm{i}rk\overline{\varphi}_\iota, 
        \label{eq:edthHalfBXi} \\
        0 &= \bar{\eth}_{-1/2}\varphi_\iota 
        \label{eq:edthHalfBarBPhi} \,\,\,\mathrm{and} \\
        0 &= \eth_{1/2}\xi_o.
        \label{eq:edthHalfAXi}
    \end{align}
    The spin-weighted spherical harmonics, $\left({}_{s}Y_{jm}\right)$, are known \cite{Penrose1984} to be eigenfunctions of $\eth_s$ and $\bar{\eth}_s$; in particular
    \begin{align}
        \eth_s\left({}_{s}Y_{jm}\right) &= \sqrt{(j - s)(j + s + 1)}\left({}_{s+1}Y_{jm}\right), \\
        \bar{\eth}_s\left({}_{s}Y_{jm}\right) &= -\sqrt{(j + s)(j - s + 1)}\left({}_{s-1}Y_{jm}\right) \,\,\,\mathrm{and} \\
        \overline{\left({}_{s}Y_{jm}\right)} &= (-1)^{s + m}\left({}_{-s}Y_{j(-m)}\right).
    \end{align}
    Furthermore, they form a complete basis for expanding functions on the round sphere. Hence, it immediately follows that the solutions to equations \ref{eq:edthHalfBarBPhi} and \ref{eq:edthHalfAXi} are 
    \begin{align}
        \varphi_\iota &= c_3\left({}_{-1/2}Y_{1/2,-1/2}\right) + c_4\left({}_{-1/2}Y_{1/2,1/2}\right) \,\,\,\mathrm{and} \\
        \xi_o &= \bar{c}_1\left({}_{1/2}Y_{1/2,-1/2}\right) + \bar{c}_2\left({}_{1/2}Y_{1/2,1/2}\right)
    \end{align}
    for some constants, $c_1$, $c_2$, $c_3$ and $c_4$.

    Substituting these into equations \ref{eq:edthHalfBarAPhi} and \ref{eq:edthHalfBXi} then says
    \begin{align}
        \bar{\eth}_{1/2}\varphi_o &= \bigg(\frac{\sqrt{2}}{\Omega}\partial_v(r)c_3 + 2\mathrm{i}krc_2\bigg)\left({}_{-1/2}Y_{1/2,-1/2}\right) \nonumber \\
        &\,\,\,\,\,\,\, + \bigg(\frac{\sqrt{2}}{\Omega}\partial_v(r)c_4 - 2\mathrm{i}krc_1\bigg)\left({}_{-1/2}Y_{1/2,1/2}\right) \,\,\,\mathrm{and} \\
        \eth_{-1/2}\xi_\iota &= \bigg(\frac{\sqrt{2}}{\Omega}\partial_u(r)\bar{c}_1 + 2\mathrm{i}kr\bar{c}_4\bigg)\left({}_{1/2}Y_{1/2,-1/2}\right) + \bigg(\frac{\sqrt{2}}{\Omega}\partial_u(r)\bar{c}_2 - 2\mathrm{i}kr\bar{c}_3\bigg)\left({}_{1/2}Y_{1/2,1/2}\right).
    \end{align}
    The claimed expressions for $\varphi_o$ and $\xi_\iota$ then follow by once again applying the completeness and eigenfunction properties (under $\eth_s$ and $\bar{\eth}_s$) of spin-weighted spherical harmonics. 
\end{proof}
The mass definition \ref{def:quasilocalMass} will be compared against is the Misner-Sharp mass \cite{Misner1964}.
\begin{definition}[Misner-Sharp mass]
    Including a cosmological constant, the Misner-Sharp mass for spherically symmetric spacetimes is defined to be
    \begin{align}
        m_{MS}(S_r^2) &= \frac{r}{2}\left(1 + 4k^2r^2 - (g^{ab} - \beta^{ab})D_a(r)D_b(r)\right),
    \end{align}
    where $\beta_{ab}$ is the induced metric on each $S_r^2$.
\end{definition}
\begin{theorem}
    $m(S_r^2)$ agrees with the Misner-Sharp mass (with cosmological constant) for spherically symmetric spacetimes.
\end{theorem}
\begin{proof}
    Taking the four $c_A$ to be the coefficients multiplying the four linearly independent solutions, it follows from theorem \ref{thm:phiSphere} that
    \begin{align}
        Q^{AB} &\equiv \frac{4r(2\partial_u(r)\partial_v(r) + \Omega^2(1 + 4k^2r^2))}{\Omega^3}\begin{bmatrix}
            \partial_ur & 0 & 0 & -\mathrm{i}k\Omega r\sqrt{2} \\
            0 & \partial_ur & \mathrm{i}k\Omega r\sqrt{2} & 0 \\
            0 & -\mathrm{i}k\Omega r\sqrt{2} & -\partial_vr & 0 \\
            \mathrm{i}k\Omega r\sqrt{2} & 0 & 0 & -\partial_vr
        \end{bmatrix}, \\
        T^{AB} &\equiv \frac{1}{\pi\Omega}\begin{bmatrix}
            0 & -\partial_ur & -\mathrm{i}k\Omega r\sqrt{2} & 0 \\
            \partial_ur & 0 & 0 & -\mathrm{i}k\Omega r\sqrt{2} \\
            \mathrm{i}k\Omega r\sqrt{2} & 0 & 0 & -\partial_vr \\
            0 & \mathrm{i}k\Omega r\sqrt{2} & \partial_vr & 0
        \end{bmatrix} \,\,\,\mathrm{and\,\,hence} \\
        T^{-1} &= \frac{\pi\Omega}{\partial_u(r)\partial_v(r) + 2k^2\Omega^2r^2}\begin{bmatrix}
            0 & \partial_vr & -\mathrm{i}k\Omega r\sqrt{2} & 0 \\
            -\partial_vr & 0 & 0 & -\mathrm{i}k\Omega r\sqrt{2} \\
            \mathrm{i}k\Omega r\sqrt{2} & 0 & 0 & \partial_ur \\
            0 & \mathrm{i}k\Omega r\sqrt{2} & -\partial_ur & 0
        \end{bmatrix}.
    \end{align}
    Then, by direct calculation, one finds 
    \begin{align}
        m(S_r^2) &= \frac{1}{16\pi}\sqrt{-\tr(QT^{-1}\overbar{Q}\overbar{T}^{-1})} =  \frac{r}{2}\left(\frac{2}{\Omega^2}\partial_u(r)\partial_v(r) + 1 + 4k^2r^2\right),
    \end{align}
    which is the Misner-Sharp mass in double null coordinates (note the Misner-Sharp mass is manifestly coordinate independent).
\end{proof}
\begin{corollary}
    For the Schwarzschild-AdS spacetime, $m(S_r^2)$ coincides with the mass parameter, $M$, in the metric.
\end{corollary}
\begin{proof}
    The Misner-Sharp mass for Schwarzschild-AdS is most easily calculated in the standard $(t, r, \theta, \phi)$ coordinates instead of double null coordinates. Hence,
    \begin{align}
        m(S_r^2) &= \frac{r}{2}\left(1 + 4k^2 + \frac{1}{1 + 4k^2 - 2M/r}\partial_t(r)^2 - (1 + 4k^2 - 2M/r)\partial_r(r)^2\right) = M
    \end{align}
    as expected.
\end{proof}
In theorem \ref{thm:qNonNegative} it was assumed that $\theta_l\theta_n < -8k^2$. However, that assumption never came up in the preceding spherical symmetry discussion. It can be shown that for the Schwarzschild-AdS spacetime,
$\theta_l\theta_n = 2(1 + 4k^2r^2 - 2M/r)$, implying $\theta_l\theta_n < -8k^2 \iff r > 2M$. This result is somewhat mysterious because $r = 2M$ is no longer a special radius when a cosmological constant is added to the Schwarzschild metric.

\subsection{Toroidal symmetry}
\label{sec:torus}
In this section, the round spheres in section \ref{sec:spherical} are replaced with tori. Because the torus is flat, it serves as the simplest - and most practically tractable - example where the constructions of section \ref{sec:definition} can be studied on a non-spherical surface, $S$. From earlier, $\mathrm{index}(\overbar{m}^a\nabla_a) = 4(1 - g)$ for a genus, $g$, surface, so there might not be any solutions to $\overbar{m}^a\nabla_a\Phi = 0$ in general. Furthermore, it's possible the number of solutions varies with sufficient deformations of the torus in question. The aim of this section is to illustrate these possibilities via counterexamples to the assumptions underpinning definition \ref{def:quasilocalMass}. First consider the $c = 0$ case in equation \ref{eq:schwarzschildGeneral}. To summarise, the domain of outer communication of the toroidal Schwarzschild-AdS spacetime is $\mathbb{R}\times((M/2k^2)^{1/3}, \infty)\times \mathbb{T}^2$ with the metric,
    \begin{align}
        g &= -f(r)^2\mathrm{d}t\otimes\mathrm{d}t + \frac{\mathrm{d}r\otimes\mathrm{d}r}{f(r)^2} + r^2(\mathrm{d}\theta\otimes\mathrm{d}\theta + \mathrm{d}\phi\otimes\mathrm{d}\phi), 
        \label{eq:kottlerMetric} \\
        \mathrm{where} \,\, f(r) &= \sqrt{-\frac{2M}{r} + 4k^2r^2}
    \end{align}
and $(\theta, \phi)$ are coordinates on each $\mathbb{T}^2 = S^1\times S^1$.

Then, one can follow the same steps as section \ref{sec:spherical}. In particular, with 
\begin{align}
        l &= \frac{1}{\sqrt{2}}\left(\frac{1}{f}\frac{\partial}{\partial t} + f\frac{\partial}{\partial r}\right),\,\, n = \frac{1}{\sqrt{2}}\left(\frac{1}{f}\frac{\partial}{\partial t} - f\frac{\partial}{\partial r}\right) \,\,\mathrm{and}\,\, m = \frac{1}{r\sqrt{2}}\left(\frac{\partial}{\partial\theta} + \mathrm{i}\frac{\partial}{\partial\phi}\right)
    \end{align}
as a NP tetrad for $\mathbb{T}_r^2$, one finds 
\begin{align}
        \sigma = \lambda = \alpha = \beta = 0 \,\,\,\mathrm{and}\,\,\, \rho = \mu = -\frac{f}{r\sqrt{2}}.
    \end{align}
Consequently $\overbar{m}^a\nabla_a\Phi = 0$ is equivalent to
    \begin{align}
        0 &= 2\partial_{\bar{z}}\varphi_o - f\varphi_\iota - 2\mathrm{i}kr\overline{\xi}_o, \\
        0 &= 2\partial_z\xi_\iota + f\xi_o + 2\mathrm{i}kr\overline{\varphi}_\iota, \\
        0 &= \partial_{\bar{z}}\varphi_\iota \,\,\,\mathrm{and} \\
        0 &= \partial_z\xi_o,
    \end{align}
where $z = \theta - \mathrm{i}\phi$. Note that $\partial_z$ is globally well defined even if $\theta$ and $\phi$ are not. By applying Liouville's theorem repeatedly and imposing periodicity in the $\theta$ and $\phi$ directions, it can be shown the general solution to these four equations is
\begin{align}
    \varphi_o = c_1, \,\,\, \xi_o = 0,\,\,\,\varphi_\iota = 0 \,\,\,\mathrm{and}\,\,\,\xi_\iota = \bar{c}_2
    \label{eq:kottlerPhi}
\end{align}
for constants, $c_1$ and $c_2$. Note that the $\Phi^A$ form of generic immediately fails because there are only two linearly independent solutions, not four. Furthermore, $\theta_l\theta_n < -8k^2$ fails to hold at any radius. By taking $c_1$ and $c_2$ to parameterise the two linearly independent solutions, direct calculation shows
\begin{align}
    Q^{AB} &= 8\sqrt{2}\pi^2r\begin{bmatrix}
            -f & -2\mathrm{i}kr \\
            2\mathrm{i}kr & -f
        \end{bmatrix}\,\,\,\mathrm{and}\,\,\,T^{AB} = 0.
\end{align}
Hence, the $T^{AB}$ form of generic also fails and $m(\mathbb{T}_r^2)$ cannot be formed via definition \ref{def:quasilocalMass}. Furthermore, even $Q^{AB}$ is not non-negative definite - an effect of the $\theta_{l, n}$ conditions failing.

It's unclear exactly what conclusions can be drawn more generally from this example. The failure of the $\Phi^A$ form of generic is unsurprising given the earlier discussion of $\mathrm{index}(\overbar{m}^a\nabla_a\Phi)$. However, the failure of the $T^{AB}$ form of generic could potentially be explained by the high degree of symmetry and flatness of the chosen surface. While deforming $\mathbb{T}_r^2$ slightly is unlikely to change the number of solutions to $\overbar{m}^a\nabla_a\Phi = 0$, it may yet result in an invertible $T^{AB}$ and a well-defined $m(S)$. It's also unclear whether there is any relation between the $\theta_{l, n}$ conditions and the generic conditions.

As the next example shows though, optimism needs to be tempered because the situation can in fact be even worse. Consider an ``AdS solition," constructed from equation \ref{eq:kottlerMetric} via the procedure in \cite{Horowitz1998}. In particular, define new coordinates, $\tau = \mathrm{i}\theta$ and $\omega = \mathrm{i}t$. Analytically continue the coordinates so that $\tau$ \& $\omega$ are real and the metric is
\begin{align}
    g &= -r^2\mathrm{d}\tau\otimes\mathrm{d}\tau + \frac{\mathrm{d}r\otimes\mathrm{d}r}{f(r)^2} + f(r)^2\mathrm{d}\omega\otimes\mathrm{d}\omega + r^2\mathrm{d}\phi\otimes\mathrm{d}\phi.
\end{align}
Unwrap the $\tau$ coordinate so $\tau \in \mathbb{R}$ and compactify the $\omega$ coordinate so that $(\omega, \phi)$ are coordinates on a torus. Avoiding a conical singularity as $r \to r_0 = (M/2k^2)^{1/3}$ forces the periodicity,
\begin{align}
    \omega \sim \omega + \frac{\pi}{3k^2r_0}.
\end{align}
Once again consider constant-$r$ tori, $\mathbb{T}_r^2$. With the NP tetrad,
\begin{align}
    l &= \frac{1}{\sqrt{2}}\left(\frac{1}{r}\frac{\partial}{\partial\tau} + f\frac{\partial}{\partial r}\right),\,\, n = \frac{1}{\sqrt{2}}\left(\frac{1}{r}\frac{\partial}{\partial\tau} - f\frac{\partial}{\partial r}\right) \,\,\mathrm{and}\,\, m = \frac{1}{\sqrt{2}}\left(\frac{1}{f}\frac{\partial}{\partial\omega} + \frac{\mathrm{i}}{r}\frac{\partial}{\partial\phi}\right),
\end{align}
one finds
\begin{align}
    \alpha = \beta = 0,\,\, \sigma = \lambda = -\frac{3M}{2\sqrt{2}r^2f} \,\,\mathrm{and}\,\, \rho = \mu = -\frac{8k^2r^3 - M}{2\sqrt{2}r^2f} = -\frac{f^2 + 12k^2r^2}{4\sqrt{2}rf}.
\end{align}
Package the GHP components of $\Phi$ into a vector, $v = (\varphi_o, \varphi_\iota, \bar{\xi}_o, \bar{\xi}_\iota)^T$. Then, with these NP coefficients, the equations of lemma \ref{thm:phiGHPEquations} become $\overbar{m}^\mu\partial_\mu v = Av$, where
\begin{align}
    A = \begin{bmatrix}
        0 & -\mu & \mathrm{i}k\sqrt{2} & 0 \\
        \bar{\sigma} & 0 & 0 & 0 \\
        0 & 0 & 0 & -\lambda \\
        0 & \mathrm{i}k\sqrt{2} & \rho & 0
    \end{bmatrix} = \frac{1}{2\sqrt{2}r^2f}\begin{bmatrix}
        0 & 8k^2r^3 - M & 4\mathrm{i}kr^2f & 0 \\ 
        -3M & 0 & 0 & 0 \\
        0 & 0 & 0 & 3M \\
        0 & 4\mathrm{i}kr^2f & -(8k^2r^3 - M) & 0
    \end{bmatrix}
\end{align}
is effectively a constant matrix on $\mathbb{T}_r^2$.

$\{m ,\overbar{m}\}$ induces a complex structure on $\mathbb{T}_r^2$. Choose a corresponding complex coordinate on the torus, $z = \frac{1}{\sqrt{2}}(f\omega - \mathrm{i}r\phi)$, so that $m^\mu\partial_\mu = \partial_z$ and $\overbar{m}^\mu\partial_\mu = \partial_{\bar{z}}$ on $\mathbb{T}_r^2$. Then, the equation to solve is $\partial_{\bar{z}}v = Av$. Integrating immediately yields 
\begin{align}
    v = \mathrm{e}^{\bar{z}A}c(z)
\end{align}
for some holomorphic vector, $c(z)$. However, by Liouville's theorem, $c(z)$ must be a constant vector, $c$. But, then $v$ would be a globally defined, non-constant, antiholomorphic vector on the compact space, $\mathbb{T}_r^2$, contradicting Liouville's theorem. The only way around this is to have $c \in \mathrm{nullspace}(A)$, so that the $\bar{z}$ dependence falls out\footnote{In fact, the toroidal Schwarzschild-AdS example earlier can be analysed in exactly this way. Since $\sigma = \lambda = 0$ in that example, the analogue of $A$ has two rows of zeroes, which then yield a 2D nullspace and the two constant solutions in equation \ref{eq:kottlerPhi}.} in $v = \mathrm{e}^{\bar{z}A}c$. However, it turns out
\begin{align}
    \mathrm{det}(A) &= (\mu^2 - 2k^2)\lambda^2.
\end{align}
Hence, by $\mu = \rho$ and lemma \ref{thm:muRho}, $A$ is invertible whenever the $\theta_l\theta_n < -8k^2$ assumption holds. In this example, this inequality always holds since
\begin{align}
    \mu^2 > 2k^2 &\iff (f^2 + 12k^2r^2)^2 > 64k^2r^2f^2 \\
    &\iff f^4 + 144k^4r^4 > 40k^2r^2f^2 \\
    &\iff 64k^2rM + \frac{4M^2}{r^2} > 0.
\end{align}
Therefore, the only solution is $v = 0$ and $\overbar{m}^a\nabla_a\Phi = 0$ has no non-trivial solutions.

The AdS soliton famously has negative energy \cite{Horowitz1998} and avoids spinorial positive energy theorems - like the one in \cite{Chrusciel2006} - because the torus has two inequivalent spin structures and the soliton's spin structure is incompatible with the one required to apply Witten's method. It's unclear if any of these properties is linked to $\overbar{m}^a\nabla_a\Phi = 0$ having only the trivial solution.

\section{Asymptotic limit}
\label{sec:asymptotics}
The next property of $m(S)$ studied is the large sphere limit. In this section, it will be convenient to set the ``AdS length scale," to 1, i.e. choose units such that $\Lambda = -3$ and $k = 1/2$. The length scales can always be restored on dimensional grounds. Note this section also features some additional indices and some slight index duplication as explained in appendix \ref{sec:conventions}.
\begin{definition}[Asymptotically AdS]
    \label{def:asymptoticallyAdS}
    A spacetime, $(M, g)$, is said to be asymptotically AdS if and only if only if the following conditions hold. $(M, g)$ must admit a conformal compactification such that conformal infinity, $\mathcal{I}$, is diffeomorphic to $\mathbb{R}\times S^2$. Then, in an open neighbourhood of $\mathcal{I}$, there must exist coordinates, $(r, x^m) = (r, t, \theta^\alpha)$, such that $\{r = \infty\}$ is $\mathcal{I}$ itself, constant $r$ \& $t$ surfaces are diffeomorphic to $S^2$ and $g$ admits a Fefferman-Graham expansion \cite{Fefferman1985},
    \begin{align}
        g &= \mathrm{e}^{2r}\bigg(-\bigg(1 + \frac{1}{4}\mathrm{e}^{-2r}\bigg)^2\mathrm{d}t\otimes\mathrm{d}t + \left(1 - \frac{1}{4}\mathrm{e}^{-2r}\right)^2g_{S^{2}} \nonumber \\
        &\,\,\,\,\,\,\,\,\,\,\,\,\,\,\,\,\,\,\,\, + \mathrm{e}^{-3r}f_{(3)mn}\,\mathrm{d}x^m\otimes\mathrm{d}x^n + O(\mathrm{e}^{-4r})\bigg) + \mathrm{d}r\otimes\mathrm{d}r,
        \label{eq:feffermanGraham}
    \end{align}
    for some symmetric tensor, $f_{(3)mn}$, that only depends on $x^m$.
\end{definition}
\begin{lemma}
    \label{thm:adsBasics}
    The metric on AdS can also be written as
    \begin{align}
        g_{\mathrm{AdS}} &= -\left(\frac{1+\rho^2}{1-\rho^2}\right)^2\mathrm{d}t\otimes\mathrm{d}t + \frac{4}{(1 - \rho^2)^2}\delta_{IJ}\mathrm{d}x^I\otimes\mathrm{d}x^J,
    \end{align}
    where $\rho = \sqrt{x^Ix_I}$. Then, with the tetrad,
    \begin{align}
    e_0 = \frac{1-\rho^2}{1+\rho^2}\partial_t \,\,\,\mathrm{and}\,\,\,e_I = \frac{1-\rho^2}{2}\partial_I,
    \label{eq:ePrime}
\end{align}
the Killing spinors\footnote{The Killing spinors are defined to be solutions of $\nabla_a\varepsilon_k = 0$.} of AdS can be written as
\begin{align}
    \varepsilon_k &= \frac{1}{\sqrt{1-\rho^2}}\left(I - \mathrm{i}x_I\gamma^I\right)\mathrm{e}^{\mathrm{i}\gamma^0t/2}\varepsilon_0,
    \label{eq:epsilonK}
    \end{align}
with $\varepsilon_0$ an arbitrary spinor that's constant with respect to the chosen tetrad.
\label{thm:epsilonK}
\end{lemma}
\begin{proof}
    AdS is written in the ``Poincar\'{e} ball" model in this lemma. To convert from the Fefferman-Graham coordinates, let $r = \ln(R + \sqrt{1 + R^2}) - \ln(2)$, where $R = \frac{2\rho}{1 - \rho^2}$ is the area-radius function. Then, with the tetrad chosen in equation \ref{eq:ePrime}, the structure equations imply the connection one-forms are
    \begin{align}
        \omega_{0I} = -\frac{2}{1+\rho^2}x_Ie^0 \,\,\,\mathrm{and}\,\,\,\omega_{IJ} = x_Je^I - x_Ie^J.
    \end{align}
    Finally, directly substituting equation \ref{eq:epsilonK} into
    \begin{align}
        \nabla_a\varepsilon_k &= e\indices{_a^\mu}\partial_\mu\varepsilon_k - \frac{1}{4}\omega_{bca}\gamma^{bc}\varepsilon_k + \frac{\mathrm{i}}{2}\gamma_a\varepsilon_k
    \end{align}
    and applying the Clifford algebra shows the claimed expression does indeed satisfy $\nabla_a\varepsilon_k = 0$. Since AdS is known to only have four linearly independent Killing spinors and equation \ref{eq:epsilonK} has four free parameters in $\varepsilon_0$, these must be the full set of solutions.
\end{proof}
\begin{definition}[``Conserved" quantities]
    \label{def:momenta}
    In any asymptotically AdS spacetime, define the energy, linear momentum, angular momentum and centre of mass position as
    \begin{align}
        E &= \frac{3}{16\pi}\int_{S_\infty^2}s^{\alpha\beta}f_{(3)\alpha\beta}\mathrm{d}(g_{S^{2}}), \\
        P_I &= \frac{3}{16\pi}\int_{S_\infty^2}s^{\alpha\beta}f_{(3)\alpha\beta}\hat{x}_I\mathrm{d}(g_{S^{2}}), \\
        J_{IJ} &= \frac{3}{16\pi}\int_{S_\infty^2}f_{(3)0\alpha}\bigg(\hat{x}_I\frac{\partial\theta^\alpha}{\partial x^J}\bigg|_{\rho = 1} - \hat{x}_J\frac{\partial\theta^\alpha}{\partial x^I}\bigg|_{\rho = 1}\bigg)\mathrm{d}(g_{S^{2}}) \,\,\,\mathrm{and} \\
        K_I &= \frac{3}{16\pi}\int_{S_\infty^2}f_{(3)0\alpha}\frac{\partial\theta^\alpha}{\partial x^J}\bigg|_{\rho = 1}\left(\delta\indices{^J_I} - \hat{x}^J\hat{x}_I\right)\mathrm{d}(g_{S^{2}})
    \end{align}
    respectively. In these expressions, $\theta^\alpha$ denote local coordinates on $S^{2}$, $s^{\alpha\beta}$ is the inverse of the round metric on $S^2$, $S_\infty^2$ is the sphere at infinity, $\hat{x}^I$ denote unit vector Cartesian coordinates and $\rho = \sqrt{x_Ix^I}$, i.e. $x^I = \rho\hat{x}^I$.
\end{definition}
These definitions are based off \cite{Chrusciel2006} - especially their equations 3.5 and 3.6 - but written in terms of Fefferman-Graham expansions/coordinates instead. Their exact form is motivated by terms that appear in positive energy theorems for asymptotically locally AdS spacetimes \cite{Rallabhandi2025}. However, some heuristics can be discussed now. It can be shown \cite{Wang2001} the Riemannian analogue of $(E, P_I)$ transforms as a Lorentz vector when one chooses a different conformal class representative to define $\mathcal{I}$. Hence, $P_I$ naturally behaves like linear momentum. Next, observe that the vector, $(\hat{x}_I\frac{\partial\theta^\alpha}{\partial x^J}|_{\rho = 1} - \hat{x}_J\frac{\partial\theta^\alpha}{\partial x^A}|_{\rho = 1})\partial_\alpha$ equals $\hat{x}_I\partial_J - \hat{x}_J\partial_I$, which is the generator of rotations. Hence, it's natural to expect the $J_{IJ}$ above to behave like angular momentum. Likewise, $\frac{\partial\theta^\alpha}{\partial x^J}|_{\rho = 1}\left(\delta\indices{^J_I} - \hat{x}^J\hat{x}_I\right)\partial_\alpha = \left(\delta\indices{^J_I} - \hat{x}^J\hat{x}_I\right)\partial_J$ can be seen as a generator of boosts when AdS is viewed as a hyperboloid in $\mathbb{R}^{3,2}$ \cite{Chrusciel2006}, suggesting the $K_I$ above should be interpreted as a centre of mass position.
\begin{theorem}
    If $S = S_\infty^2$, i.e. the sphere at infinity in an asymptotically AdS spacetime, then $Q(\Phi) = Q(\varepsilon_k)$, where $\varepsilon_k$ is a Killing spinor of AdS.
\end{theorem}
\begin{proof}
    First note the $\Phi^A$ form of generic holds in this instance because $S_r^2$ becomes increasingly round as $r \to \infty$ and spherically symmetric spacetimes are known to satisfy this notion of generic from section \ref{sec:schwarzschild}. AdS itself has four globally defined, linearly independent solutions to $\nabla_a\Phi = 0$, namely the the 4D space of Killing spinors, $\varepsilon_k$. Therefore in AdS, the space of solutions to $\overbar{m}^a\nabla_a\Phi = 0$ on any generic surface can be spanned by simply restricting the Killing spinors to the surface.

    By definition \ref{def:asymptoticallyAdS}, the difference between $g$ and $g_{\mathrm{AdS}}$ is $O(\mathrm{e}^{-3r})$. Hence, in the asymptotic region of $(M, g)$, $\Phi = \varepsilon_k + \mathcal{Z}$ for some $\mathcal{Z}$ that is $O(\mathrm{e}^{-3r})$ below leading order. Equating $\frac{4}{(1-\rho^2)^2}\mathrm{d}\rho\otimes\mathrm{d}\rho$ with $\mathrm{d}r\otimes\mathrm{d}r$ in lemma \ref{thm:epsilonK} shows $\varepsilon_k$ is $O(\mathrm{e}^{r/2})$ and thus $\mathcal{Z}$ must be $O(\mathrm{e}^{-5r/2})$.

    In vierbein indices $P_a = -\delta_{a0}$ and $Q_a = \delta_{a1} \equiv \mathrm{d}r$. Therefore,
    \begin{align}
        Q(\Phi) &= \int_{S_\infty^2}E^{01}(\Phi)\mathrm{d}A \\
        &= Q(\varepsilon_k) + \int_{S_\infty^2}\big(\mathcal{Z}^\dagger\gamma^1\gamma^A\nabla_A\varepsilon_k + \varepsilon_k^\dagger\gamma^1\gamma^A\nabla_A\mathcal{Z} + \mathcal{Z}^\dagger\gamma^1\gamma^A\nabla_A\mathcal{Z} + \nabla_A(\varepsilon_k)^\dagger\gamma^A\gamma^1\mathcal{Z} \nonumber \\
        &\,\,\,\,\,\,\, + \nabla_A(\mathcal{Z})^\dagger\gamma^A\gamma^1\varepsilon_k + \nabla_A(\mathcal{Z})^\dagger\gamma^A\gamma^1\mathcal{Z}\big)\mathrm{d}A.
    \end{align}
    From equation \ref{eq:feffermanGraham}, $\mathrm{d}A$ is $O(\mathrm{e}^{2r})$. Consequently, $\mathcal{Z}^\dagger\gamma^1\gamma^A\nabla_A(\mathcal{Z})\mathrm{d}A$ and $\nabla_A(\mathcal{Z})^\dagger\gamma^A\gamma^1\mathcal{Z}\mathrm{d}A$ are both $O(\mathrm{e}^{-3r})$ and go to zero as $r \to \infty$. Meanwhile, it can be shown \cite{Rallabhandi2025} $\nabla_A\varepsilon_k$ is also $O(\mathrm{e}^{-5r/2})$ in the asymptotic region, implying $\mathcal{Z}^\dagger\gamma^1\gamma^A\nabla_A(\varepsilon_k)\mathrm{d}A$, $\nabla_A(\varepsilon_k)^\dagger\gamma^A\gamma^1\mathcal{Z}\mathrm{d}A \to 0$ too. That leaves
    \begin{align}
        Q(\Phi) &= Q(\varepsilon_k) + \int_{S_\infty^2}\big(\varepsilon_k^\dagger\gamma^1\gamma^A\nabla_A\mathcal{Z} + \nabla_A(\mathcal{Z})^\dagger\gamma^A\gamma^1\varepsilon_k\big)\mathrm{d}A.
        \label{eq:randoEquation2}
    \end{align}
    The 2nd term is the complex conjugate of the first so it suffices to prove the 1st term integrates to zero. This term can be re-written as
    \begin{align}
        \varepsilon_k^\dagger\gamma^1\gamma^A\nabla_A\mathcal{Z} &= D_A(\varepsilon_k^\dagger\gamma^1\gamma^A\mathcal{Z}) - \nabla_A(\varepsilon_k)^\dagger\gamma^1\gamma^A\mathcal{Z}.
    \end{align}
    From above, $\nabla_A(\varepsilon_k)^\dagger\gamma^1\gamma^A\mathcal{Z}$ contributes nothing to the integral, leaving
    \begin{align}
        \int_{S_\infty^2}\varepsilon_k^\dagger\gamma^1\gamma^A\nabla_A\mathcal({Z})\mathrm{d}A &= \int_{S_\infty^2}D_A(\varepsilon_k^\dagger\gamma^1\gamma^A\mathcal{Z})\mathrm{d}A.
    \end{align}
    Let $D^{(S)}_A$ be the intrinsic Levi-Civita connection of $S$, let $K_{IJ}$ be the extrinsic curvature of $\Sigma$ in $M$ and let $c_{AB}$ be the extrinsic curvature of $S$ in $\Sigma$. Then,
    \begin{align}
        D_A(\varepsilon_k^\dagger\gamma^1\gamma^A\mathcal{Z}) &= \left(D_A^{(S)}\varepsilon_k - \frac{1}{2}K_{AI}\gamma^I\gamma^0\varepsilon_k - \frac{1}{2}c_{AB}\gamma^B\gamma^1\varepsilon_k\right)^\dagger\gamma^1\gamma^A\mathcal{Z} \nonumber \\
        &\,\,\,\,\,\,\, + \varepsilon_k^\dagger\gamma^1\gamma^A\left(D_A^{(S)}\mathcal{Z} - \frac{1}{2}K_{AI}\gamma^I\gamma^0\mathcal{Z} - \frac{1}{2}c_{AB}\gamma^B\gamma^1\mathcal{Z}\right) \\
        &= D_A^{(S)}(\varepsilon_k^\dagger\gamma^1\gamma^A\mathcal{Z}) + \frac{1}{2}K_{AI}(\varepsilon_k^\dagger\gamma^0\gamma^I\gamma^1\gamma^A\mathcal{Z} - \varepsilon_k^\dagger\gamma^1\gamma^A\gamma^I\gamma^0\mathcal{Z}),
    \end{align}
    where the extrinsic curvature's symmetry is used to cancel the $c_{AB}$ terms. The measure, $\mathrm{d}A$, is $O(\mathrm{e}^{2r})$ while the $\varepsilon_k$-$\mathcal{Z}$ products are already $O(\mathrm{e}^{-2r})$. Since AdS is time-symmetric, $K_{IJ} = 0$ to leading order, meaning the second term contributes nothing to the integral. That leaves
    \begin{align}
        \int_{S_\infty^2}\varepsilon_k^\dagger\gamma^1\gamma^A\nabla_A\mathcal({Z})\mathrm{d}A &= \int_{S_\infty^2}D_A^{(S)}(\varepsilon_k^\dagger\gamma^1\gamma^A\mathcal{Z})\mathrm{d}A = 0
    \end{align}
    by Stokes' theorem. Equation \ref{eq:randoEquation2} then implies $Q(\Phi) = Q(\varepsilon_k)$.
\end{proof}
\begin{corollary}
    When $S = S_\infty^2$, $m(S) = \sqrt{E^2 - ||P||^2 + ||J||^2 - ||K||^2}$.
    \label{thm:massLimit}
\end{corollary}
\begin{proof}
    Based on similar methods to \cite{Chrusciel2006}, it can be shown \cite{Rallabhandi2025}
    \begin{align}
        Q(\varepsilon_k) &= 8\pi\varepsilon_0^\dagger\mathrm{e}^{-\mathrm{i}\gamma^0t/2}\bigg(EI - \mathrm{i}P_I\gamma^I + K_I\gamma^0\gamma^I + \frac{\mathrm{i}}{2}J_{IJ}\gamma^0\gamma^{IJ}\bigg)\mathrm{e}^{\mathrm{i}\gamma^0t/2}\varepsilon_0.
    \end{align} 
    The four components of the constant spinor, $\varepsilon_0$, can be used to parameterise the four linearly independent solutions, $\Phi^A$. Hence, as a matrix,
    \begin{align}
        Q^{AB} &\equiv 8\pi\mathrm{e}^{-\mathrm{i}\gamma^0t/2}\bigg(EI - \mathrm{i}P_I\gamma^I + K_I\gamma^0\gamma^I + \frac{\mathrm{i}}{2}J_{IJ}\gamma^0\gamma^{IJ}\bigg)\mathrm{e}^{\mathrm{i}\gamma^0t/2}.
    \end{align}
    The other matrix required for $m(S_\infty^2)$ is $T^{AB}$. In the present context, it will be easiest to use the alternative expression, $T^{AB} = (\Phi^A)^TC^{-1}\Phi^B$, of equation \ref{eq:tDirac}. In the conventions chosen, $(\gamma_a)^T = -C^{-1}\gamma_aC$. Therefore, 
    \begin{align}
        T^{AB} &= (\varepsilon_k^A)^TC^{-1}\varepsilon_k^B \\
        &= \frac{1}{1 - \rho^2}(\varepsilon_0^A)^T\left(\mathrm{e}^{\mathrm{i}\gamma^0t/2}\right)^T\left(I - \mathrm{i}x_I(\gamma^I)^T\right)C^{-1}\left(I - \mathrm{i}x_J\gamma^J\right)\mathrm{e}^{\mathrm{i}\gamma^0t/2}\varepsilon_0^B \\
        &= \frac{1}{1 - \rho^2}(\varepsilon_0^A)^TC^{-1}\mathrm{e}^{-\mathrm{i}\gamma^0t/2}C\left(I + \mathrm{i}x_IC^{-1}\gamma^IC\right)C^{-1}\left(I - \mathrm{i}x_J\gamma^J\right)\mathrm{e}^{\mathrm{i}\gamma^0t/2}\varepsilon_0^B \\
        &= (C^{-1})^{AB}.
    \end{align}
    Finally, with the $Q^{AB}$ and $T^{AB}$ derived, definition \ref{def:quasilocalMass} reduces to
    \begin{align}
        m(S_\infty^2) &= \frac{1}{16\pi}\sqrt{-\tr(QT^{-1}\overbar{Q}\overbar{T}^{-1})} = \sqrt{E^2 - P_IP^I + J_IJ^I - K_IK^I},
    \end{align}
    where $J_I = \frac{1}{2}\varepsilon_{IJK}J^{JK}$.
\end{proof}
The question naturally arises whether $\sqrt{E^2 - ||P||^2 + ||J||^2 - ||K||^2}$ is an appropriate notion of mass in asymptotically AdS spacetimes. For example, from special relativity, one thinks of mass as just $\sqrt{E^2 - ||P||^2}$, without any contributions from angular momenta, $J_{IJ}$, or boost charges, $K_I$. However, it can be argued this is an artefact of Minkowski space's symmetry group, namely the Poincar\'{e} group. As in quantum field theory, one could define $m^2$ to be proportional to a quadratic Casimir operator of (the Lie algebra of) the symmetry group \cite{Buchbinder1998}. Therefore, in the AdS context, one seeks a quadratic Casimir for $\mathfrak{o(3, 2)}$.

Choose generators, $\{J_{MN} = -J_{NM}\}_{M, N = 1}^5$, for $\mathfrak{o(3, 2)}$ such that the defining Lie bracket is\footnote{The fact such a basis exists can be seen immediately by following the analogous steps in \cite{Weinberg1995} for $\mathfrak{o(3, 1)}$.}
    \begin{align}
        \left[J^{MN}, J^{PQ}\right] &= \mathrm{i}\left(\eta^{MP}J^{NQ} - \eta^{MQ}J^{NP} - \eta^{NP}J^{MQ} + \eta^{NQ}J^{MP}\right),
    \end{align}
where $\eta_{MN} \equiv \mathrm{diag}(-1, 1, 1, 1, -1)$ and all $M, N, \cdots$ indices are raised/lowered by $\eta^{-1}$/$\eta$. Then, it immediately follows that $C = \frac{1}{2}J^{MN}J_{MN}$ is a quadratic Casimir\footnote{Assume we are working with a faithful matrix representation of the Lie algebra so that multiplying two Lie algebra elements is well-defined.} for $\mathfrak{o(3, 2)}$.

Interpret $J^{5a}$ as a 4-momentum generator, $\mathbb{P}^a$, $J^{0I}$ as boost generators, $\mathbb{K}^I$, and $J^{IJ}$ as angular momentum generators, $\mathbb{J}_I = \frac{1}{2}\varepsilon_{IJK}J^{JK}$, in line with \cite{Chrusciel2006} and the heuristics accompanying definition \ref{def:momenta}. Then,
\begin{align}
    C &= J^{5a}J_{5a} + \frac{1}{2}J^{IJ}J_{IJ} + J^{0I}J_{0I} = \mathbb{P}^0\mathbb{P}^0 - \mathbb{P}^I\mathbb{P}_I + \mathbb{J}^I\mathbb{J}_I - \mathbb{K}^I\mathbb{K}_I, 
\end{align}
suggesting that the limit in corollary \ref{thm:massLimit} is physically reasonable.

\section{Linearised gravity}
\label{sec:linearisation}
This section concerns perturbations of AdS sourced by a matter field. In particular, the metric is assumed to be 
\begin{align}
    g_{ab} &= B_{ab} + \eta h_{ab},
\end{align}
where $B = g_{\mathrm{AdS}}$ is the background metric, $h$ is the perturbation and $\eta$ is assumed to be an infinitesimal parameter. Furthermore, the energy-momentum tensor, $T_{ab}$, is assumed to be $O(\eta)$. The aim is to show that definition \ref{def:quasilocalMass} captures a mass associated with $T_{ab}$. Throughout this section, the coordinates and tetrad will be the same as in lemma \ref{thm:adsBasics}. It will once again be convenient to set the AdS length scale to one, i.e. choose units where $\Lambda = -3$ and $k = 1/2$.

A natural way to construct physical quantities, like mass, out of the energy-momentum tensor is to contract $T_{ab}$ with the Killing vectors of the background metric. It can be checked that the Killing vectors of AdS are spanned by
    \begin{align}
        \tau &= \partial_t, \\
        j_{IJ} &= x_I\partial_J - x_J\partial_I, \\
        p_I &= \frac{2x_I}{1 + \rho^2}\cos(t)\partial_t + \frac{1}{2}\left((1 + \rho^2)\delta\indices{^J_I} - 2x^Jx_I\right)\sin(t)\partial_J \,\,\,\,\mathrm{and} \\
        k_I &= -\frac{2x_I}{1 + \rho^2}\sin(t)\partial_t + \frac{1}{2}\left((1 + \rho^2)\delta\indices{^J_I} - 2x^Jx_I\right)\cos(t)\partial_J.
    \end{align}
In analogy with definition \ref{def:momenta}, one can then define the following ``matter charges."
\begin{definition}[Matter charges]
    \label{def:matterCharges}
    Let matter charges on $\Sigma$ be defined as
    \begin{align}
        E &= \int_\Sigma T_{0a}\tau^a\mathrm{d}V, \,\,\,
        P_I = \int_\Sigma T_{0a}p_I^a\mathrm{d}V, \,\,\, 
        J_{IJ} = \int_\Sigma T_{0a}j_{IJ}^a\mathrm{d}V \,\,\,\mathrm{and} \,\,\,
        K_I = \int_\Sigma T_{0a}k_I^a\mathrm{d}V.
    \end{align}
\end{definition}
\begin{theorem}
    For gravity linearised about AdS, if $S$ is generic in the $\Phi^A$ sense, then
    \begin{align}
        m(S) &= \sqrt{E^2 - ||P||^2 + ||J||^2 - ||K||^2}.
    \end{align}
\end{theorem}
\noindent This expression is formally identical to corollary \ref{thm:massLimit} and therefore the result can once again be thought of as a Casimir mass, but this time for $T_{ab}$.
\begin{proof}
    AdS already has four linearly independent Killing spinors, i.e. global solutions to $\nabla_a^{(B)}\varepsilon_k = 0$. Hence, in AdS, if $S$ is generic in the $\Phi^A$ sense, then solutions to $\overbar{m}^a\nabla_a\Phi = 0$ can be spanned by simply restricting $\varepsilon_k$ to $S$. Since $g_{ab} = B_{ab} + \eta h_{ab}$ though, one can therefore let $\Phi = \varepsilon_k + \eta\mathcal{Z}$ for some Dirac spinor, $\mathcal{Z}$.

    Extend $\mathcal{Z}$'s definition off $S$ in an arbitrary, but sufficiently regular, way so that $\Phi = \varepsilon_k + \eta\mathcal{Z}$ is defined on all of $\Sigma$. Then, by definition \ref{def:q}, 
    \begin{align}
        Q(\Phi) &= 2\int_\Sigma\left(\nabla_I(\Phi)^\dagger\nabla^I\Phi - 4\pi T^{0a}\overline{\Phi}\gamma_a\Phi - (\gamma^I\nabla_I\Phi)^\dagger\gamma^J\nabla_J\Phi\right)\mathrm{d}V.
        \label{eq:randoEquation3}
    \end{align}
    $\nabla_a^{(B)}\varepsilon_k = 0$ implies $\nabla_a\Phi = O(\eta)$ and thus the first and third terms in equation \ref{eq:randoEquation3} are both $O(\eta^2)$. Meanwhile, since $T_{ab}$ is assumed to be $O(\eta)$, the second term is $-4\pi T^{0a}\bar{\varepsilon}_k\gamma_a\varepsilon_k + O(\eta^2)$.
    \newline    In summary, the linearised limit yields,
    \begin{align}
        Q(\Phi) &\to 8\pi\int_\Sigma T_{0a}\bar{\varepsilon}_k\gamma^a\varepsilon_k\mathrm{d}V.
        \label{eq:qPhiLinearised}
    \end{align}
    Direct calculation using equation \ref{eq:epsilonK} shows
        \begin{align}
        \bar{\varepsilon}_k\gamma^0\varepsilon_k
        &= \frac{1}{1-\rho^2}\varepsilon_0^\dagger\mathrm{e}^{-\mathrm{i}\gamma^0t/2}((1 + \rho^2)I - 2\mathrm{i}x_I\gamma^I)\mathrm{e}^{\mathrm{i}\gamma^0t/2}\varepsilon_0 \,\,\,\mathrm{and} \\
        \bar{\varepsilon}_k\gamma^I\varepsilon_k
        &= \frac{1}{1-\rho^2}\varepsilon_0^\dagger\mathrm{e}^{-\mathrm{i}\gamma^0t/2}((1 + \rho^2)\gamma^0\gamma^I - 2\mathrm{i}x_J\gamma^0\gamma^{IJ} - 2x^Ix_J\gamma^0\gamma^J)\mathrm{e}^{\mathrm{i}\gamma^0t/2}\varepsilon_0.
    \end{align}
    Substituting back into equation \ref{eq:qPhiLinearised}, applying definition \ref{def:matterCharges} and converting to vierbein indices where required then gives
    \begin{align}
        Q(\Phi) &\to 8\pi\varepsilon^\dagger_0\bigg(\int_\Sigma\frac{1 + \rho^2}{1 - \rho^2}T_{00}\mathrm{d}V\,I - 2\mathrm{i}\int_\Sigma\frac{x_I}{1 - \rho^2}T_{00}\mathrm{e}^{-\mathrm{i}\gamma^0t}\mathrm{d}V\,\gamma^I - 2\mathrm{i}\int_\Sigma\frac{x_J}{1 - \rho^2}T_{0I}\mathrm{d}V\,\gamma^0\gamma^{IJ} \nonumber \\
        &\,\,\,\,\,\, + \int_\Sigma\frac{1}{1 - \rho^2}T_{0I}\left((1 + \rho^2)\delta\indices{^I_J} - 2x^Ix_J\right)\mathrm{e}^{-\mathrm{i}\gamma^0t}\mathrm{d}V\,\gamma^0\gamma^{J}\bigg)\varepsilon_0 \\
        &= 8\pi\varepsilon^\dagger_0\bigg(\int_\Sigma\frac{1 + \rho^2}{1 - \rho^2}T_{00}\mathrm{d}V\,I + \mathrm{i}\int_\Sigma\frac{1}{1 - \rho^2}(x_IT_{0J} - x_JT_{0I})\mathrm{d}V\,\gamma^0\gamma^{IJ} \nonumber \\
        &\,\,\,\,\,\,\, - \mathrm{i}\int_\Sigma\left(\frac{2x_I\cos(t)}{1 - \rho^2}T_{00} + \frac{\sin(t)}{1 - \rho^2}T_{0I}\left((1 + \rho^2)\delta\indices{^I_J} - 2x^Ix_J\right)\right)\mathrm{d}V\,\gamma^I \nonumber \\
        &\,\,\,\,\,\,\, + \int_\Sigma\left(-\frac{2x_I\sin(t)}{1 - \rho^2}T_{00} + \frac{\cos(t)}{1 - \rho^2}T_{0I}\left((1 + \rho^2)\delta\indices{^I_J} - 2x^Ix_J\right)\right)\mathrm{d}V\,\gamma^0\gamma^I\bigg)\varepsilon_0 \\
        &= 8\pi\varepsilon_0^\dagger\left(EI + \frac{\mathrm{i}}{2}J_{IJ}\gamma^0\gamma^{IJ} - \mathrm{i}P_I\gamma^I + K_I\gamma^0\gamma^I\right)\varepsilon_0.
    \end{align}
    Taking the four components of the constant spinor, $\varepsilon_0$, to parameterise the four linearly independent solutions, $\Phi^A$, then gives
    \begin{align}
        Q^{AB} &\to 8\pi\left(EI + \frac{\mathrm{i}}{2}J_{IJ}\gamma^0\gamma^{IJ} - \mathrm{i}P_I\gamma^I + K_I\gamma^0\gamma^I\right).
    \end{align}
    Since this $Q^{AB}$ is already $O(\eta)$, for the linearised limit it suffices to take $T^{AB}$ to $O(1)$ in definition \ref{def:quasilocalMass}. Thus, $T^{AB} \to (\varepsilon_k^A)^TC^{-1}\varepsilon_k^B = (C^{-1})^{AB}$, borrowing the calculation from the proof of corollary \ref{thm:massLimit}.

    Finally, by evaluating $m(S)$ for this $Q^{AB}$ and $T^{AB}$, one finds
    \begin{align}
        m(S) &= \frac{1}{16\pi}\sqrt{-\tr(QT^{-1}\overbar{Q}\overbar{T}^{-1})} \to \sqrt{E^2 + J_IJ^I - P_IP^I - K_IK^I},
    \end{align}
    where $J_I = \frac{1}{2}\varepsilon_{IJK}J^{JK}$.
\end{proof}

\section{Conclusion}
\label{sec:conclusion}
In this work, a new quasilocal mass has been defined for generic surfaces in spacetimes with negative cosmological constant. The new definition is spinorial and based on definitions by Penrose and Dougan \& Mason - which are themselves related to Witten's proof of the positive energy theorem. The new definition has been shown to satisfy a number of physically desirable properties - namely that $m(S) \geq 0$, $m(S) = 0$ for every generic surface in AdS, $m(S_r^2)$ agrees with the Misner-Sharp mass in spherical symmetry and $m(S)$ has an appropriate limit, $\sqrt{E^2 - ||P||^2 + ||J||^2 - ||K||^2}$, for gravity linearised about AdS or as $S$ approaches a sphere on $\mathcal{I}$ in an asymptotically AdS spacetime.

Some avenues of further research are immediately apparent at this juncture. This work was originally inspired by Reall's suggestion that a quasilocal mass-charge inequality could be established for spacetimes with $\Lambda < 0$ and such an inequality could be used to prove the 3rd law of black hole mechanics for supersymmetric horizons in this context. Having now established a workable quasilocal mass, a logical next step would be tackling this conjecture, likely using a modified connection inspired by the gravitino transformation in 4D, minimal, gauged supergravity. Note that even if a quasilocal mass-charge inequality is found, the 3rd law might not be immediately accessible because the $\theta_l\theta_n < -8k^2$ requirement prevents taking $S$ arbitrarily close to the event horizon (where $\theta_l = 0$).

Even for the definition of quasilocal mass itself, some improvements could be made. Two different definitions of generic were given in this work and it may be interesting to study further how the two definitions relate. It would be particularly desirable to find an example with toroidal $S$ where the quasilocal mass construction can be carried out in full, unlike the examples in section \ref{sec:torus}. Then, perhaps a more concrete conclusion can be made about whether either definition is generic in practice or physically relevant for higher genus surfaces.

Elsewhere, in terms of physical properties, one property not mentioned in this work is the ``small sphere" limit. In particular, it's often insisted \cite{Szabados2009} that given a point, $p \in M$, and a future direction, $t^a$, if $S_r$ is a sphere reached by flowing an affine parameter distance, $r$, along the generators of $p$'s future lightcone, then a quasilocal 4-momentum for $S_r$ should be $P^a = -\frac{4\pi}{3}r^3T\indices{^a_b}t^b$, to leading order in $r$. Then, $m(S_r)^2 = -P^aP_a$. This happens to be true for both the Dougan-Mason and Penrose masses \cite{Dougan1992, Kelly1986}. For a vacuum spacetime a similar result should hold at 5th order in $r$ based on the Bel-Robinson tensor. It would be interesting to see if this - or something analogous - also holds for the definition proposed in this work.

Furthermore, while the new definition has good quantitative properties, it does share some qualitative failings of many other quasilocal masses. Unlike the Hamilton-Jacobi masses \cite{Brown1993, Kijowski1997, Wang2009} or the Hawking mass \cite{Hawking1968}, the physical motivation for this definition is not clear, beyond some supergravity considerations \cite{Horowitz1983} underpinning Witten's method to prove the positive energy theorem. More practically, like the Hamilton-Jacobi or Bartnik \cite{Bartnik1989} masses, this quasilocal mass is likely to be quite difficult to calculate for most metrics - not only does one have to find a NP tetrad adapted to $S$, one then has to find all solutions to $\overbar{m}^a\nabla_a\Phi = 0$ on $S$. The quest to find a truly satisfactory quasilocal mass goes on.

Another possible extension would be to consider spacetimes with $\Lambda > 0$ instead. Not only is the $\Lambda > 0$ case potentially most relevant to the real world, it is arguably also a pressing need for mathematical general relativity. Many familiar properties of conformal infinity break down when $\Lambda > 0$ \cite{Ashtekar2014} and this precludes defining anything directly analogous to the ADM \cite{Arnowitt1962} or Wang \cite{Wang2001} masses. Nonetheless, a number of energy-momentum definitions have been devised in this context - see \cite{Szabados2019} for a review. Particularly relevant to the present work are extensions based on Witten's method \cite{Kastor2002, Szabados2015}. Ultimately though, these successes still have to work around the global challenges imposed by $\Lambda > 0$ - e.g. compact Cauchy surfaces, spacelike $\mathcal{I}^+$ or cosmological horizons. It may be that quasilocal mass is a viable alternative for avoiding these issues. In fact, an analogue of Penrose's quasilocal mass can be defined for asymptotically de Sitter spacetimes, albeit it no longer retains some key properties, such as positivity \cite{Szabados2015, Tod2015}. Likewise, it would be interesting to see if the definition developed in this article can be adjusted for $\Lambda > 0$ and if so, which physical properties remain intact.

\section{Acknowledgements}
Most of all, I would like to thank my supervisor, James Lucietti, for discussions and guidance throughout this project. I would also like to thank Tim Adamo, Harvey Reall and Paul Tod for reading drafts of this paper and providing helpful comments. Finally, I would like to thank the University of Edinburgh's School of Mathematics for my PhD studentship funding. 

\begin{appendices}
\section{Conventions}
\label{sec:conventions}
The conventions used in this work are based off \cite{Buchbinder1998}; the main points are listed below.
\newline
\newline
The metric signature is $(-1, +1, +1, +1)$.
\newline
\newline
The following symbols are frequently used.
\begin{itemize}
    \item $M$: The full spacetime
    \item $g$: The Lorentzian metric on $M$
    \item $\Sigma$: 3D, compact, spacelike hypersurface with boundary
    \item $S$: The boundary of $\Sigma$
    \item $\Lambda$: A negative cosmological constant
    \item $k = \sqrt{-\Lambda/12}$
    \item $C_b^\infty$: The space of smooth Dirac spinors on $\Sigma$ subject to the boundary conditions given in definition \ref{def:cbInfinity}
    \item $\mathcal{H}$: The completion of $C_b^\infty$ under the inner product in lemma \ref{def:cbInnerProduct}
    \item $\mathfrak{D}$: Operator from $\mathcal{H} \to L^2$ defined by $\Psi \mapsto \gamma^I\nabla_I\Psi$
    \item $\overline{\Psi} = \Psi^\dagger\gamma^0$ for a Dirac spinor, $\Psi$
    \item $D_a$: The Levi-Civita connection of $g$
    \item $\nabla_a\Psi = D_a\psi + \mathrm{i}k\gamma_a\Psi$ for a Dirac spinor, $\Psi$
    \item $\nabla_a\overline{\Psi} = D_a\overbar{\Psi} - \mathrm{i}k\overline{\Psi}\gamma_a = (\nabla_a\Psi)^\dagger\gamma^0$ for a Dirac spinor, $\Psi$
    \item $I$: The identity matrix
    \item $\{o_\alpha,\, \iota_\alpha\}$: A GHP spinor dyad
    \item $\delta = m^aD_a$ in the context of the NP formalism
    \item $\bar{\delta} = \overbar{m}^aD_a$ in the context of the NP formalism
\end{itemize}
Many different types of indices are used, as given below.
\begin{itemize}
    \item $a, b, c, \cdots$ are vierbein indices running $0, 1, 2, 3$. However, in most equations it will be apparent that these could equally well denote abstract indices.
    \item $\mu, \nu, \rho, \cdots$ are coordinate indices running $0, 1, 2, 3$.
    \item $I, J, K, \cdots$ are vierbein indices running $1, 2, 3$.
    \item $\alpha, \beta, \gamma, \cdots$ are two-component spinor indices for the $(1/2, 0)$ representation, i.e. left-handed Weyl spinors, and run 1, 2.
    \item $\dot{\alpha}, \dot{\beta}, \dot{\gamma}, \cdots$ are two-component spinor indices for the $(0, 1/2)$ representation, i.e. right-handed Weyl spinors, and run $\dot{1}$, $\dot{2}$.
    \item $A, B, C, \cdots$ index the linearly independent solutions to $\overbar{m}^a\nabla_a\Phi = 0$.
\end{itemize}
The Riemann tensor is defined such that $[D_a, D_b]V^c = R\indices{^c_d_a_b}V^d$.
\newline
\newline
Complex conjugation of an object - unless it is a Dirac spinor - will be denoted by a bar over the object, e.g. $\bar{z}$.
\newline
\newline
Levi-Civita symbols are normalised by $\varepsilon_{12} = -1$, $\varepsilon^{12} = 1$, $\varepsilon_{\dot{1}\dot{2}} = -1$, $\varepsilon^{\dot{1}\dot{2}} = 1$, $\varepsilon_{0123} = -1$ and $\varepsilon^{0123} = 1$. Then, $\varepsilon^{\alpha\gamma}\varepsilon_{\gamma\beta} = \delta\indices{^\alpha_\beta}$ and likewise for the dotted indices.
\newline
\newline
Two-component spinors are raised and lowered from the left, i.e. $\psi_\alpha = \varepsilon_{\alpha\beta}\psi^\beta$ and $\psi^\alpha = \varepsilon^{\alpha\beta}\psi_\beta$.
\newline
\newline
The extended Pauli matrices are
\begin{align}
    (\sigma_a)_{\alpha\dot{\alpha}} &\equiv (I, \sigma_1, \sigma_2, \sigma_3) \,\,\,\mathrm{and} \\
    (\tilde{\sigma}_a)^{\dot{\alpha}\alpha} &= \varepsilon^{\alpha\beta}\varepsilon^{\dot{\alpha}\dot{\beta}}(\sigma_a)_{\beta\dot{\beta}} \equiv (I, -\sigma_1, -\sigma_2, -\sigma_3),
\end{align}
with $\sigma_{1, 2, 3}$ being the standard Pauli matrices.
\newline
\newline
The conversion between vierbein indices and two-component spinor indices is by $V_{\alpha\dot{\alpha}} = (\sigma_a)_{\alpha\dot{\alpha}}V^a$ and $V_a = -\frac{1}{2}(\tilde{\sigma}_a)^{\dot{\alpha}\alpha}V_{\alpha\dot{\alpha}}$.
\newline
\newline
Dirac spinors are decomposed into two-component spinors by $\Psi = (\psi_\alpha, \bar{\chi}^{\dot{\alpha}})^T$.
\newline
\newline
Gamma matrices are in the Weyl representation, i.e.
\begin{align}
    \gamma_a &= \begin{bmatrix} 0 & (\sigma_a)_{\alpha\dot{\alpha}} \\
    (\tilde{\sigma}_a)^{\dot{\alpha}\alpha} & 0 \end{bmatrix}.
\end{align}
Hence, the gamma matrices are unitary and satisfy $\gamma^a\gamma^b + \gamma^b\gamma^a = -2g^{ab}I$. Furthermore, in terms of two-component spinors, $\overline{\Psi} = \Psi^\dagger\gamma^0 = (-\chi^\alpha, -\overline{\psi}_{\dot{\alpha}})$.
\newline
\newline
The antisymmetric product, $\gamma^{[a_1}\cdots\gamma^{a_n]}$, is denoted by $\gamma^{a_1\cdots a_n}$.
\newline
\newline
The charge conjugation matrix is
\begin{align}
    C &= \begin{bmatrix}
        \varepsilon_{\alpha\beta} & 0 \\
        0 & \varepsilon^{\dot{\alpha}\dot{\beta}}
    \end{bmatrix} \iff C^{-1} = \begin{bmatrix}
        \varepsilon^{\alpha\beta} & 0 \\
        0 & \varepsilon_{\dot{\alpha}\dot{\beta}}
    \end{bmatrix}.
\end{align}
\newline
The spin-weighted spherical harmonics used in section \ref{sec:schwarzschild} are
\begin{align}
    \left({}_{1/2}Y_{1/2,1/2}\right) &= \frac{\mathrm{i}}{\sqrt{2\pi}}\sin\left(\frac{\theta}{2}\right)\mathrm{e}^{\mathrm{i}\phi/2}, \,\,\, \left({}_{1/2}Y_{1/2,-1/2}\right) = -\frac{\mathrm{i}}{\sqrt{2\pi}}\cos\left(\frac{\theta}{2}\right)\mathrm{e}^{-\mathrm{i}\phi/2}, \nonumber \\
    \left({}_{-1/2}Y_{1/2,1/2}\right) &= \frac{\mathrm{i}}{\sqrt{2\pi}}\cos\left(\frac{\theta}{2}\right)\mathrm{e}^{\mathrm{i}\phi/2} \,\,\, \mathrm{and}\,\,\,\left({}_{-1/2}Y_{1/2,-1/2}\right) = \frac{\mathrm{i}}{\sqrt{2\pi}}\sin\left(\frac{\theta}{2}\right)\mathrm{e}^{-\mathrm{i}\phi/2}.
\end{align}
\newline
Section \ref{sec:asymptotics} features some additional or modified index conventions as listed below.
\begin{itemize}
    \item Based on context, $\alpha, \beta, \gamma, \cdots$ also denote coordinate indices running $2, 3$.
    \item Based on context, $A, B, C, \cdots$ also denote vierbein indices running $2, 3$.
    \item $m, n, p, \cdots$ are coordinate indices running $0, 2, 3$.
    \item $M, N, P, \cdots$ run $0, 1, 2, 3, 4$ and index the embedding Cartesian coordinates when AdS is viewed as a surface in $\mathbb{R}^{3, 2}$.
\end{itemize}

\subsection{Comparison to Penrose-Rindler conventions}
\label{sec:penroseRindler}
The monogrpahs of Penrose and Rindler \cite{Penrose1984, Penrose1986} (PR) are popular references for two-component spinors. However, their conventions differ at times from the ones used in this article; the key differences are listed below.
\begin{itemize}
    \item A mostly plus metric is used here while PR use a mostly minus metric.
    \item Lowercase letters from the start of the Greek alphabet are used for two-component spinor indices here while PR use uppercase Latin letters.
    \item For two-component spinors, the undotted indices here correspond to their primed indices and the dotted indices here correspond to their unprimed indices.
    \item The two-component spinor indices here run over the values 1 and 2, whereas PR take them to run over 0 and 1.
    \item The spinor index conversion here is $V_{\alpha\dot{\alpha}} = (\sigma_a)_{\alpha\dot{\alpha}}V^a$, while PR define $V_{\alpha\dot{\alpha}} = \frac{1}{\sqrt{2}}(\sigma_a)_{\alpha\dot{\alpha}}V^a$.
    \item The $\sqrt{2}$ difference when converting to spinor indices means $\iota^\alpha o_\alpha = \bar{\iota}^{\dot{\alpha}}\bar{o}_{\dot{\alpha}} = \sqrt{2}$ here while PR have $\iota^{A^\prime}o_{A^\prime} = \iota^Ao_A = 1$.
    \item For any spinor, $\psi_\alpha$, the $\psi_o$ and $\psi_\iota$ in equation \ref{eq:psiOPsiI} would be $\frac{1}{\sqrt{2}}\psi_{1^\prime}$ and $-\frac{1}{\sqrt{2}}\psi_{0^\prime}$ respectively in their notation. Likewise, $\overline{\psi}_o$ and $\overline{\psi}_\iota$ are $\frac{1}{\sqrt{2}}\psi_{1}$ and $-\frac{1}{\sqrt{2}}\psi_{0}$ respectively.
    \item Dirac spinors are $\Psi = (\psi_\alpha,\, \overbar{\chi}^{\dot{\alpha}})^T$ here while PR would write $(\overbar{\chi}^A,\, \psi_{A^\prime})^T$, i.e. the left and right handed componenets are written in the opposite order.
    \item Indices are raised and lowered from the left here, i.e. $\psi^\alpha = \varepsilon^{\alpha\beta}\psi_\beta$ and $\psi_\alpha = \varepsilon_{\alpha\beta}\psi^\beta$, while PR raise from the left, but lower from the right, i.e. $\psi^A = \varepsilon^{AB}\psi_B$ but $\psi_A = \psi^B\varepsilon_{BA}$. This difference means $\varepsilon^{12} = 1 \implies \varepsilon_{12} = -1$ here. Furthermore, it means $\varepsilon^{\alpha\gamma}\varepsilon_{\gamma\beta} = \delta\indices{^\alpha_\beta}$ here while they have $\varepsilon^{AC}\varepsilon_{CB} = -\delta\indices{^A_B}$.
\end{itemize}

\section{Frequently used spinor identities}
\label{sec:identities}
The following are some two-component spinor identities used liberally - most are from \cite{Buchbinder1998}.
\begin{align}
    V^{\alpha\dot{\alpha}}W_{\alpha\dot{\alpha}} &= -2V^aW_a \\
    \overline{(\psi_\alpha)} &= \overline{\psi}_{\dot{\alpha}} \\
    \psi_\alpha\chi^\alpha &= - \psi^\alpha\chi_\alpha \\
    (\sigma_a)_{\alpha\dot{\alpha}}(\tilde{\sigma}_b)^{\dot{\alpha}\beta} + (\sigma_b)_{\alpha\dot{\alpha}}(\tilde{\sigma}_a)^{\dot{\alpha}\beta} &= -2g_{ab}\delta\indices{_\alpha^\beta} \\
    (\tilde{\sigma}_a)^{\dot{\alpha}\alpha}(\sigma_b)_{\alpha\dot{\beta}} + (\tilde{\sigma}_b)^{\dot{\alpha}\alpha}(\sigma_a)_{\alpha\dot{\beta}} &= -2g_{ab}\delta\indices{^{\dot{\alpha}}_{\dot{\beta}}} \\
    (\sigma_a)_{\alpha\dot{\alpha}}(\tilde{\sigma}_b)^{\dot{\alpha}\alpha} &= -2g_{ab} \\
    (\sigma^a)_{\alpha\dot{\alpha}}(\tilde{\sigma}_a)^{\dot{\beta}\beta} &= -2\delta\indices{^\beta_\alpha}\delta\indices{^{\dot{\beta}}_{\dot{\alpha}}} \\
    (\sigma_a)_{\alpha\dot{\beta}}(\tilde{\sigma}_b)^{\dot{\beta}\beta}(\sigma_c)_{\beta\dot{\alpha}} &= g_{ca}(\sigma_b)_{\alpha\dot{\alpha}} - g_{bc}(\sigma_a)_{\alpha\dot{\alpha}} - g_{ab}(\sigma_c)_{\alpha\dot{\alpha}} + \mathrm{i}\varepsilon_{abcd}(\sigma^d)_{\alpha\dot{\alpha}} \\
    (\tilde{\sigma}_a)^{\dot{\alpha}\beta}(\sigma_b)_{\beta\dot{\beta}}(\tilde{\sigma}_c)^{\dot{\beta}\alpha} &= g_{ca}(\tilde{\sigma}_b)^{\dot{\alpha}\alpha} - g_{bc}(\tilde{\sigma}_a)^{\dot{\alpha}\alpha} - g_{ab}(\tilde{\sigma}_c)^{\dot{\alpha}\alpha} - \mathrm{i}\varepsilon_{abcd}(\tilde{\sigma}^d)^{\dot{\alpha}\alpha} \\
    \varepsilon_{\alpha\beta}\varepsilon^{\gamma\delta} &= -\left(\delta\indices{^\gamma_\alpha}\delta\indices{^\delta_\beta} - \delta\indices{^\delta_\alpha}\delta\indices{^\gamma_\beta}\right)
\end{align}
The following NP coefficients are needed in terms of GHP variables.
\begin{align}
    \sqrt{2}\mu &= \sqrt{2}\overbar{m}^a\delta n_a = \bar{\iota}^{\dot{\alpha}}\delta\bar{\iota}_{\dot{\alpha}} \\
    \sqrt{2}\rho &= -\sqrt{2}m^a\bar{\delta}l_a = \bar{o}^{\dot{\alpha}}\bar{\delta}\bar{o}_{\dot{\alpha}} \\
    \sqrt{2}\alpha &= \frac{1}{\sqrt{2}}\left(\overbar{m}^a\bar{\delta}m_a - n^a\bar{\delta}l_a\right) = \bar{\iota}^{\dot{\alpha}}\bar{\delta}\bar{o}_{\dot{\alpha}} \\
    \sqrt{2}\beta &= \frac{1}{\sqrt{2}}\left(\overbar{m}^a\delta m_a - n^a\delta l_a\right) = \bar{\iota}^{\dot{\alpha}}\delta\bar{o}_{\dot{\alpha}} \\
    \sqrt{2}\sigma &= -\sqrt{2}m^a\delta l_a = \bar{o}^{\dot{\alpha}}\delta\bar{o}_{\dot{\alpha}} \\
    \sqrt{2}\lambda &= \sqrt{2}\overbar{m}^a\bar{\delta}n_a = \bar{\iota}^{\dot{\alpha}}\bar{\delta}\bar{\iota}_{\dot{\alpha}}
\end{align}

\end{appendices}

\bibliographystyle{unsrt}
\bibliography{references}

\end{document}